\documentclass[reqno]{amsart}
\usepackage{amsmath,amssymb,cite}
\usepackage[mathscr]{euscript}
\usepackage{mathtools}
\mathtoolsset{showonlyrefs}
\usepackage{xcolor}

\usepackage{ulem}

\definecolor{bblue}{rgb}{.2,0.2,.8}

\theoremstyle{plain}
\newtheorem{theorem}{Theorem}[section]

\newtheorem{lemma}[theorem]{Lemma}

\theoremstyle{definition}

\theoremstyle{remark}
\newtheorem{remark}[theorem]{Remark}
\newtheorem{example}[theorem]{Example}
\newcommand{\pare}[1]{\left(#1\right)}%comando per le parentesi

\numberwithin{equation}{section}
\numberwithin{theorem}{section}

\def\be{\begin{equation}}
\def\ee{\end{equation}}
\def\bp{\begin{pmatrix}}
\def\ep{\end{pmatrix}}
\def\bea{\begin{eqnarray}}
\def\eea{\end{eqnarray}}

\def\\{\par\medskip}

\let\0=\noindent

\newcommand{\id}{{1 \mskip -5mu {\rm I}}}
\renewcommand{\epsilon}{\varepsilon}
\renewcommand{\hat}{\widehat}

\title[Mass exchange models]{ On a class of solvable stationary non equilibrium states for mass exchange models }

\author{M. Capanna}
\address{\noindent Monia Capanna \hfill\break\indent
	Banca d'Italia, 
	\hfill\break\indent
	Via Otricoli, 41, 00181 Roma, Italy
}
\email{monia.capanna@bancaditalia.it}

\author{D. Gabrielli}
\address{\noindent Davide Gabrielli \hfill\break\indent
 DISIM, Universit\`a dell'Aquila
\hfill\break\indent
Via Vetoio Loc. Coppito, 67100 L'Aquila, Italy
}
\email{davide.gabrielli@univaq.it}

\author{D. Tsagkarogiannis}
\address{\noindent Dimitrios Tsagkarogiannis \hfill\break\indent
	DISIM, Universit\`a dell'Aquila
	\hfill\break\indent
	Via Vetoio Loc. Coppito, 67100 L'Aquila, Italy
}
\email{dimitrios.tsagkarogiannis@univaq.it}

\begin{document}

\begin{abstract}
We consider a family of models having an arbitrary positive amount of mass on each site and randomly exchanging an arbitrary amount of mass with nearest neighbor sites. We restrict to the case of diffusive models. We identify a class of reversible models for which the product invariant measure is known and the gradient condition is satisfied so that we can explicitly compute the transport coefficients associated to the diffusive hydrodynamic rescaling. Based on the Macroscopic Fluctuation Theory \cite{mft} we have that the large deviations rate functional for a stationary non equilibrium state can be computed solving a Hamilton-Jacobi equation depending only on the transport coefficients and the details of the boundary sources. Thus, we are able to identify a class of models having transport coefficients for which the Hamilton-Jacobi equation can indeed be solved. We give a complete characterization in the case of generalized zero range models and discuss several other cases. For the generalized zero range models we identify a class of discrete models that, modulo trivial extensions, coincides with the class discussed in \cite{FG} and a class of continuous dynamics that coincides with the class in \cite{FFG}. Along the discussion we obtain a complete characterization of reversible misanthrope processes solving the discrete equations in \cite{CC}.

\end{abstract}

\noindent
\keywords{Stochastic particle systems, transport coefficients, stationary non equilibrium states}

\subjclass[2010]
{Primary
60K35, %Interacting random processes; statistical mechanics type
       %models; percolation theory
60F10; %Large deviations
Secondary
82C22, %Interacting particle systems in time-dependent statistical mechanics
82C70. %Transport processes in time-dependent statistical mechanics
}

\maketitle
\thispagestyle{empty}

\section{Introduction}
Among the various challenges in non-equilibrium thermodynamics is the understanding of stationary non equilibrium states (SNS). In the last years stochastic interacting particle systems played a crucial role in this direction. This is due to the fact that they are effective toy models for which, sometimes, even exact computations are possible.
In particular, an important breakthrough has been the exact computation of the large deviations rate functional for an one-dimensional SNS of some special stochastic interacting system \cite{mft,D}.

The prototype of such solvable models is the simple exclusion process \cite{Li}. For such a model in the one-dimensional case it is possible to give a combinatorial representation of the invariant measure in terms of a product of matrices. From this, it is possible to deduce the large deviations rate functional \cite{DLS},
which is written in terms of a nonlinear boundary value problem and it is non local, a fact that reflects the presence of long range correlations. Similar combinatorial arguments can be developed for the asymmetric version of the model \cite{DLS-a1,DLS-a2}.

An alternative approach is the dynamical variational one as presented in 
the macroscopic fluctuation theory \cite{mft}.
The large deviations rate functional for the SNS of boundary driven exclusion process can be computed as the quasipotential of the dynamic large deviations rate functional from the hydrodynamic limit \cite{BDGJL1,BDGJL2}. This approach can be applied to general models with bulk dynamics that is gradient and reversible. Within this approach the large deviations rate functionals are only related to the macroscopic properties of the model, in particular the transport coefficients, that is the diffusion matrix and the mobility. In particular, an explicit computation of a non local large deviations rate functional for the SNS is possible for any one-dimensional model having a constant diffusion coefficient and a mobility that is a second order polynomial in the density. Such an example is the Kipnis-Marchioro-Presutti (KMP) model and its dual \cite{BGL}.

Consequently, a problem of interest is to construct models whose bulk dynamics is at the same time gradient and reversible and for which the diffusion coefficient is constant and the mobility is quadratic. This is in general a difficult task \cite{Spohn}. For models of this type, since the large deviations rate functional for the SNS can be computed macroscopically, we expect that relatively simple and interesting microscopic descriptions of the invariant measures could be found.

in this paper we investigate models for which on each site there can be an arbitrary amount of mass, either discrete or continuous, and in a single jump an arbitrary amount of mass may jump. We were inspired by other models of this type found in the literature (for example \cite{EMSZ,RS}) and especially by the papers \cite{FG,FFG} where the authors compute all moments of the SNS for a class of boundary driven one-dimensional generalized zero range dynamics.
Note that for the classic zero range dynamics and a few other models \cite{DMF,DMOP,EH}, even in the boundary driven case, no long range correlations appear, while instead the generalized versions in \cite{FG,FFG} have long range correlations.

We develop a systematic analysis of the models that are at the same time gradient and reversible and such that they have constant diffusion and quadratic mobility. We consider models of zero range type as
well as models for which the rate of jump depends on the occupation variables both of the starting and the arrival site, like in \cite{CC}. In particular, our main result, Theorem~\ref{teozr}, is a complete classification in the case of zero range dynamics. According to this we obtain, modulo some trivial generalizations, that there is a whole family of discrete as well as continuous models that are exactly coinciding with those in \cite{FG} and \cite{FFG}, respectively. In \cite{FG,FFG} the models were introduced and characterized by motivations coming from quantum spin systems; we characterize the same classes in a different way and moreover establish that there are no other zero range models having this integrability property.
Along the way we also give a complete solution to the discrete equations in \cite{CC}.

A motivation for our analysis is related to the conjecture in \cite{BDGJLstat}: for solvable SNS with quadratic and convex mobility the non local large deviation rate functional can be interpreted as the result of the contraction principle of hidden variables; this means that we may expect that even microscopically the SNS may have a simple description by introducing some hidden variables. This has been shown to be true in the KMP model and for its dual in \cite{DMFG}.
For a zero range dynamics with quadratic mobility we can have only a convex mobility and this means that the generalized zero range models in Theorem \ref{teozr} are natural candidates to verify the conjecture in \cite{BDGJLstat}. This has indeed been obtained in \cite{CFGGT} for some specific cases where the representation is simpler and then generalized
to a larger class of models in \cite{CFFGR}.

\smallskip

The structure of the paper is the following. 

\noindent In Section~\ref{s:models} we introduce the basic concepts and notation to handle
generalized models where an arbitrary amount of mass may jump and then
we identify general classes of models that are both reversible and gradient 
of the following type: zero range (Section~\ref{sec:zr}), KMP (Section~\ref{sec:kmp})
and Misanthrope (Section~\ref{sec:mis}) for both the discrete and the continuous case. In particular we give the general solution of the discrete equations in \cite{CC}.

\noindent In Section~\ref{sec:TC} after a short introduction about some challenges for proving the hydrodynamic limit and large deviations, we show how the transport coefficients, namely the diffusion and the mobility, are computed.

\noindent The main results can be found in Section~\ref{sec:ZRD} where in Theorem~\ref{teozr} we identify exactly the class of zero range processes that will have a computable non local large deviations rate functional for the SNS and in Section~\ref{sec:other} where we compute the transport coefficients for some cases of the KMP-type and Misanthrope-type models we considered earlier.

\section{Models}\label{s:models}
Let $\Lambda_N$ be a graph and call $\mathcal E_N$ the collection of unoriented edges. We fix as reference a canonical orientation; i.e., we fix an
oriented graph with vertices $\Lambda_N$ and oriented edges $E_N$ such that for each unoriented edge $\{x,y\}\in \mathcal E_N$ there exists only one corresponding fixed orientation, either $(x,y)$ or $(y,x)$, that belongs to $E_N$. This fixed orientation does not influence the dynamics and it is fixed just for notational convenience.

We consider a configuration of masses $\eta\in \mathbb R_+^{\Lambda_N}$ associated to the nodes of the lattice. The value $\eta(x)\in \mathbb R_+$ represents the amount of mass present at $x\in \Lambda_N$.
In the following $\Lambda_N$ will almost always be a one dimensional lattice but we think it is useful to define the models on an arbitrary graph.

\subsection{Invariant measures}
We consider a family of product measures on $\mathbb R_+^{\Lambda_N}$ having one dimensional marginal given by
\begin{equation}\label{mis}
g_\lambda(d\eta(x)):=\frac{e^{\lambda\eta(x)}}{Z(\lambda)}g(d\eta(x))\,, \quad x\in \Lambda_N\,,\quad Z(\lambda):=\int_{\mathbb R_+}g(d\eta)e^{\lambda\eta},
\end{equation}
where $\lambda$ is a parameter that represents the chemical potential and $Z(\lambda)$
is a normalization factor.

The chemical potential $\lambda$ modulates the typical density associated to the measure. We have indeed that the function
\begin{equation}\label{robeta}
\rho(\lambda)=\int_{\mathbb R_+} g_\lambda (d\eta)\eta=\frac{Z'(\lambda)}{Z(\lambda)}=\frac{d}{d\lambda} \log Z(\lambda)\,,
\end{equation}
is monotone increasing. We have that \eqref{robeta} represents the typical density corresponding to the value $\lambda$ of the chemical potential; we call $\lambda(\rho)$ the corresponding inverse function. In general the probability measure $g_\lambda$ may not be defined for all values of $\lambda$. We have indeed a well defined probability measure for $\lambda<\lambda_c$, where $\lambda_c$ may coincide with $+\infty$. The critical value $\lambda_c$ is defined by
\begin{equation}\label{lambdac}
\lambda_c:=\inf\left\{\lambda\,,\, Z(\lambda)=+\infty\right\}\,.
\end{equation}
In the models we will consider we assume that
$\lim_{\lambda\uparrow \lambda_c}\rho(\lambda)=+\infty$, in such a way that varying the chemical potential we can obtain any possible typical value of the density.

The measure $g(d\eta)$ in \eqref{mis} is a generic positive measure on $\mathbb R^+$. We will consider both cases when the positive measure $g$ is absolutely continuous with respect to the Lebesgue measure, i.e., when $g(d\eta)=g(\eta)d\eta$  as well as when $g$ is a purely atomic measure, i.e., it is of the form
$g(d\eta)=\sum_{i=1} g\left(\eta^{(i)}\right) \delta_{\eta^{(i)}}(d\eta)$ for some configurations $\eta^{(i)}$. With a slight abuse of notation we call again $g$ both the density and the weights of the delta measures.

\begin{example}
	An important example is the absolutely continuous family of measures having density $g\equiv 1$. In this case
	$\lambda_c=0$, the measure $g_\lambda$ is an exponential distribution with parameter $|\lambda|$ and $Z(\lambda)=\rho(\lambda)=1/|\lambda|$.
\end{example}

We will often write statements using the notation of the absolutely continuous case  but the results hold also in the purely atomic case with a simple generalization of the formulation. In special cases we will explicitly discuss the purely atomic case too.

\subsection{Conservative exchange dynamics}
We consider a general conservative exchange dynamics on $\mathbb R_+^{\Lambda_N}$ where mass can be exchanged between nearest neighbor sites. The general form of the generator is given by
\begin{equation}\label{generatore}
\mathcal L_Nf(\eta)=\sum_{\{x,y\}\in \mathcal E_N}\int r_{\{x,y\}}(\eta,d \eta')\left[f(\eta')-f(\eta)\right],
\end{equation}
where $r_{\{x,y\}}(\eta,\cdot)$ is a positive measure on $\mathbb R_+^{\Lambda_N}$ concentrated on the configurations $\eta'$ such that $\eta'(z)=\eta(z)$ when $z\neq x,y$ and $\eta(x)+\eta(y)=\eta'(x)+\eta'(y)$.
We assume that $r_{\{x,y\}}(\eta,\cdot)$ depends on $\eta$ just by the values $\eta(x)$ and $\eta(y)$.

\smallskip

We will mainly consider regular lattices with a stochastic mechanism that is the same on each bond $\{x,y\}$ so that there exists a positive measure  $r(\zeta,d\zeta')$ on $\mathbb R_+^2$, with also $\zeta\in \mathbb R_+^2$, such that $r_{\{x,y\}}$ is obtained by $r$ with $\zeta=(\zeta_1,\zeta_2)=(\eta(x),\eta(y))$ and $\eta'(z)=\eta(z)$ for any $z\neq x,y$ and
$\eta'(x)=\zeta'_1$ and $\eta'(y)=\zeta'_2$. Here $(x,y)\in E_N$ is the oriented bond corresponding to $\{x,y\}\in \mathcal E_N$. We will always consider cases when the measure $r$ is symmetric in the simultaneous exchange of the indices $1,2$ in $\zeta$ and $\zeta'$, but in principle this constraint could be removed for models with a directional asymmetry.

Note that the support of the measure $r$ is contained in the one dimensional subspace identified by
$S:=\left\{\zeta':\zeta_1+\zeta_2=\zeta'_1+\zeta'_2\right\}$.

We do not discuss mathematical details of the generator \eqref{generatore} and the definition of the corresponding stochastic dynamics for which we could possibly add
proper boundary conditions. We stress that this class of models is very large and contain very different models that may exhibit condensation and may have pure jump paths.
We will only briefly discuss these features.

\subsection{Current}
Since the measure $r$ is concentrated on a one dimensional subset of $\mathbb R_+^2$, a convenient alternative way of describing the evolution is in terms of the measure $R(\zeta, \cdot)$ on $\mathbb R$ denoting the rate at which a given current is observed across the bond $(x,y)$ (a similar approach has been considered in \cite{DCG}). When the configuration is changing from $\eta$ to $\eta'$, with the constraints discussed in the previous section, the current observed on $(x,y)$ is given by
$\eta(x)-\eta'(x)$. The measure $R(\zeta, dj)$ with $\zeta=(\zeta_1,\zeta_2)$ has then support in $[-\zeta_2,\zeta_1]$. When the measures $r$ and $R$ can be described by densities we simply have $R(\zeta,j)=r(\zeta,\zeta_1-j,\zeta_2+j)$; a similar relation holds in the case of jumps of discrete values of the mass.

We will mainly consider the situation when there is not a preferred direction on the system that is summarized by the symmetry relation
\begin{equation}\label{simcard}
R((\zeta_1,\zeta_2),j)=R((\zeta_2,\zeta_1),-j)\,;
\end{equation}
the meaning of this relation is that if you exchange the masses at the extremes of an edge the observed current is reversed.

Denoting by $\delta^x$ the element of $\mathbb R^{\Lambda_N}_+$ such that $\delta^x_x=1$ and $\delta^x_y=0$
for any $y\neq x$, the generator \eqref{generatore} can be therefore written as
\begin{equation}\label{generatoreR}
\mathcal L_Nf(\eta)=\sum_{(x,y)\in E_N}\int_{-\eta(y)}^{\eta(x)} R(\eta(x),\eta(y),d j)\left[f(\eta-j\delta^x+j\delta^y)-f(\eta)\right]\,.
\end{equation}

\begin{remark}\label{remark}
Another natural viewpoint of the dynamics is in terms of flows. We will use several times this perspective.
A flow is always positive and describes the amount of mass transferred from one vertex to another one. When an amount of mass $q$ is transferred from $x$ to $y$ we say that there is a flow $q\geq 0$ from $x$ to $y$. Assuming the validity of the symmetry \eqref{simcard} we can describe the rates of flows by one single positive measure
$Q(\zeta, \cdot)$ on $\mathbb R_+$ with support on $[0,\zeta_1]$. The relations with the measure $R$ written in the absolutely continuous case is given by
\begin{equation}\label{sp-flows}
R(\zeta_1,\zeta_2,j)=Q(\zeta_1,\zeta_2,j)\id(j\geq 0)+Q(\zeta_2,\zeta_1,-j)\id(j< 0)\,.
\end{equation}
In the general case we should use two measures $Q^\pm$ respectively for positive and negative currents $j$ in the above formula.
\end{remark}

\subsection{Reversibility}

We search for measures $R$ such that the dynamics with generator \eqref{generatoreR}
is reversible with respect to the product measure $\prod_{x\in \Lambda_N}g_\lambda(d\eta(x))$.
By the conservative nature of the dynamics this is independent from the parameter $\lambda$ and the detailed balance condition is given by
\begin{equation}\label{detbal2}
g(\zeta_1)g(\zeta_2)R(\zeta,j)=g(\zeta'_1)g(\zeta'_2)R(\zeta',-j)\,,\qquad j\in [\zeta_1,-\zeta_2]\,,
\end{equation}
where $\zeta'_1=\zeta_1-j$ and $\zeta'_2=\zeta_2+j$. For notational convenience it is useful to introduce
$\tilde R(\zeta, j):=g(\zeta_1)g(\zeta_2)R(\zeta,j)$.
The following elementary lemma classifies completely the models that satisfy \eqref{detbal2}.
\begin{lemma}\label{Rtilde}
A model satisfies \eqref{detbal2} if and only if there exists an arbitrary function $\psi(\zeta,j)\geq 0, j\geq 0$ such that
the function $\tilde R(\zeta, j)$ introduced above
has the form:
\begin{equation}\label{fissoRtilde}
\tilde R(\zeta, j)=\left\{
\begin{array}{ll}
\psi(\zeta, j)\,, & j\geq 0\\
\psi(\zeta_1-j,\zeta_2+j, -j)\,, & j\leq 0\,.
\end{array}
\right.
\end{equation}
\end{lemma}
\begin{proof}
The fact that $\psi(\zeta, j)$ can be arbitrary for positive $j$ follows by the fact that in \eqref{detbal2}
the currents on the two sides have always opposite sign and there are no constraints on currents having the same sign. The form of $\tilde R$ for negative $j$ in \eqref{fissoRtilde} is determined by the fact that \eqref{detbal2} should be satisfied for any value of the current.
\end{proof}
Using the above lemma one can easily construct all the dynamics that have an invariant measure determined by the function $g$; indeed there is a whole family depending on an arbitrary function
$\psi(\zeta,j)\geq 0, j\geq 0$. Once this function is fixed, the rates of the dynamics are given by
\begin{equation}\label{fissoR}
R(\zeta, j)=\left\{
\begin{array}{ll}
\frac{\psi(\zeta, j)}{g(\zeta_1)g(\zeta_2)}\,, & j\geq 0\\
\frac{\psi(\zeta_1-j,\zeta_2+j, -j)}{g(\zeta_1)g(\zeta_2)}\,, & j< 0\,.
\end{array}
\right.
\end{equation}
In order to have the symmetry \eqref{simcard}, the arbitrary function $\psi$ has to satisfy the relation
\begin{equation}
\psi(\zeta_1,\zeta_2,j)=\psi(\zeta_2+j,\zeta_1-j, j)\,, \qquad j\geq 0\,.
\end{equation}

\subsection{Instantaneous current and gradient condition}

Given a bond $(x,y)\in E_N$  and a configuration $\eta\in\mathbb R_+^{\Lambda_N}$ then the instantaneous current $j_\eta(x,y)$ is the average of the current flowing with respect to the positive measure given by the rate. By construction we will have $j_\eta(x,y)=-j_\eta(y,x)$. The formula for the definition of the instantaneous current is therefore the following
\begin{equation}
j_\eta(x,y):=\int_{-\eta(y)}^{\eta(x)}R(\eta(x),\eta(y),j)jdj\,, \qquad (x,y)\in E_N\,.
\end{equation}
Recall that for simplicity of notation we consider the case where $R$ is absolutely continuous.
The mathematical motivation of such a definition is that, by the general theory of Markov processes \cite{Spohn}, we have that if we call $\mathcal Q(x,y,t)$
the amount of mass flown from $x$ to $y$ up to time $t$ and define by
\begin{equation}\label{defJ}
J(x,y,t):=\mathcal Q(x,y,t)-\mathcal Q(y,x,t)
\end{equation}
the net current, we have that
\begin{equation}\label{defmart}
J(x,y,t)-\int_0^tj_{\eta_s}(x,y)ds
\end{equation}
is a martingale that plays a crucial role in the study of the limiting behavior of the system in the scaling limit as it will be briefly discussed in Section~\ref{HL}. Note that $(\eta_t)_{t\geq 0}$ is the process generated by \eqref{generatoreR}.

\smallskip

We restrict now to the case of $\Lambda_N=\mathbb Z^d/(N\mathbb Z)^d$, the $d$ dimensional discrete torus and we denote by $(\tau_x)_{x\in \Lambda_N}$ the group of translations.
The model is called of \emph{gradient type} if there exists a local function $h(\eta)$ such that
\begin{equation}\label{grcond-giu}
j_\eta(x,y)=\tau_xh(\eta)-\tau_yh(\eta)\,.
\end{equation}
Since we are assuming that the rate of jump across the bond $\{x,y\}$ depends on $\eta$ just by the values $\eta(x)$ and $\eta(y)$ then we can deduce that $h$ depends only on the value of $\eta$ on one single site.
To prove this we need a few definitions.

\smallskip
Given $D\subseteq \Lambda_N$ we denote by $\eta_D:=\left(\eta(x)\right)_{x\in D}$ the configuration $\eta$ restricted to the subset $D$. Given $\eta, \zeta$ two configurations we denote by $\eta_D\zeta_{D^c}$
the configuration that coincides with $\eta$ on $D$ and with $\zeta$ on $D^c$ (where $D^c$ denotes the complementary of $D$).

Given a function $h:\left(\mathbb R^+\right)^{\Lambda_N}\to \mathbb R$, its domain of dependence $D\subseteq \Lambda_N$ is defined by requiring the following two properties:

1) $h\left(\zeta_{D^c}\eta_D\right)=h\left(\eta\right)$ for any $\zeta$ and $\eta$,

2) for any $x\in D$ there exist a configuration $\eta$ and a real parameter $\alpha$ such that
$h(\eta+\alpha\delta^x)\neq h(\eta)$.

We call a function local if its domain of dependence is finite and does not depend on $N$.

\begin{lemma}
Given a process with generator as in \eqref{generatoreR}, if there is a local function $h$ satisfying \eqref{grcond-giu} then its domain of dependence is $D=\{0\}$ and therefore $h(\eta)=H(\eta(0))$, for a suitable real function $H$.
\end{lemma}
\begin{proof}
First of all note that, given $j_\eta$, there is not a unique $h$ satisfying \eqref{grcond-giu}. Indeed, considering $h,h'$ be two solutions of \eqref{grcond-giu} we obtain that
$\tau_x(h-h')=\tau_y(h-h')$ which in turn implies that $h-h'$ is translationally invariant. Indeed the whole set of solutions is given by $h+t$ where $h$ is a fixed solution and $t$ is a generic translation invariant function.

We concentrate now on local functions. For simplicity let us consider the one dimensional case (but the argument works in any dimension) with $y=x+e_1$ and suppose that the domain of dependence of the local function $h$ is given by an interval $D=[a,b]$. We have then that the domain of dependence of the local function on the left hand side of \eqref{grcond-giu} is $[x,x+e_1]$ while the domain of dependence of the right hand side of \eqref{grcond-giu} contains $\{a+x,b+x+e_1\}$ and this is a contradiction unless $a=b=0$.
\end{proof}

The previous analysis yields that the gradient condition corresponds to having a real function $H$ such that the first moment of the
measure $R$ can be written as
\begin{equation}\label{grad}
\int_{-\zeta_2}^{\zeta_1}R(\zeta,j)jdj=H(\zeta_1)-H(\zeta_2)\,.
\end{equation}
For the models that we consider it is rather simple to
select models satisfying the gradient condition.

\begin{example}[Zero range models]\label{zre}
This class of models will play an important role.
We call \emph{Zero Range models} the models for which the rate $Q(\zeta, \cdot)$ associated to the flow (as discussed in Remark \ref{remark}) depends only on the mass present at the starting point, i.e., $Q$ does not depend on $\zeta_2$. The corresponding rate for the current is therefore
\begin{equation}\label{Rzr}
R(\zeta,j)=Q(\zeta_1,j)\id(j\geq 0)+Q(\zeta_2,-j)\id(j<0)\,.
\end{equation}
Hence, we have
\begin{equation*}
\int_{-\zeta_1}^{\zeta_2}R(\zeta,j)jdj=\int_0^{\zeta_1}Q(\zeta_1,j)jdj-\int_0^{\zeta_2}Q(\zeta_2,j)jdj
\end{equation*}
and the gradient condition is satisfied with
the function $H$ given by
$H(a)=\int_0^aQ(a,j)jdj$.
\end{example}

\begin{example}[Additive mean zero terms]
Condition \eqref{grad} continues to be satisfied with the same function $H$ if we add a mean zero term to the rate $R$. Indeed, in this way we obtain all possible gradient models associated to a given function $H$.
Consider a gradient rate $R$ with respect to a function $H$ and add $R'$ such that $\int_{-\zeta_2}^{\zeta_1}R'(\zeta,j)j dj=0$; then $R+R'$ is gradient with respect to the same function $H$.

\smallskip

As an example we can add to a zero range rate (or to any gradient model) a term of the form
\begin{equation}\label{tutte?}
R'(\zeta,j)=S(\zeta,j)\id\left(-(\zeta_1 \wedge \zeta_2)\leq j\leq (\zeta_1 \wedge \zeta_2) \right)\,,
\end{equation}
where $S(\zeta,j)=S(\zeta,-j)$ and $a\wedge b:= \min\{a,b\}$, getting again a gradient model with the same function $H$. The only constraint that has to be satisfied is that $R+R'\geq 0$, which is always satisfied if $S\geq 0$.

\smallskip
Another possible choice is to add to $R$ a term $R'$ of the form
\begin{equation}\label{consing}
R'(\zeta, j)=\frac{A\left(\zeta, j-\frac{\zeta_1-\zeta_2}{2}\right)}{j}\,,
\end{equation}
where $A(\zeta,\cdot)$ is an antisymmetric function. Here again we have to impose a positivity condition to the total rate, and we have to carefully avoid the possible singularity for $j=0$ in \eqref{consing}. In particular if we do not want an explosion on the rate and if we have a continuous $A$ we need to impose $A\left(\zeta,\frac{\zeta_2-\zeta_1}{2}\right)=0$.

\end{example}

\subsection{Reversible+Gradient}\label{sec:zr}

We look for gradient models that are at the same time reversible with respect to the product measure with marginals $g_\lambda$. This is because for models of this type it is possible to have an explicit form of the transport coefficients. Using Lemma \ref{Rtilde}, after simple change of variables, we get that we need to find a function $\psi (\zeta, j)\geq 0, j\geq 0$ such that
\begin{equation}\label{condpsi}
\int_0^{\zeta_1}\psi(\zeta,j)jdj-\int_0^{\zeta_2}\psi(\zeta_1+j,\zeta_2-j,j)jdj=g(\zeta_1)g(\zeta_2)\Big(H(\zeta_1)-H(\zeta_2)\Big)\,,
\end{equation}
for a suitable real function $H$. Other equivalent formulations are possible but in any case it is difficult to get a clear complete classification of all the models satisfying the two conditions.
We discuss here instead some remarkable examples, finding a large and interesting class of models. We start with the Zero range already presented in Example~\ref{zre}:

\begin{example}[Zero range]\label{zre2}
The family of zero range models has been discussed in Example~\ref{zre}. They are always of gradient type and in order to satisfy the reversible condition for an arbitrary product measure with marginals $g_\lambda$,  imposing \eqref{detbal2}, we obtain after a simple computation the following special form for the rates
\begin{equation}\label{zrrates}
R(\zeta,j)=\frac{g(\zeta_1-j)}{g(\zeta_1)}S(|j|)\id(j\geq 0)+\frac{g(\zeta_2+j)}{g(\zeta_2)}S(|j|)\id(j < 0),
\end{equation}
where $S$ is an arbitrary function (details of the computation are simple and similar to computations in the examples to follow). Relation \eqref{zrrates} can be naturally written in terms of the rate $Q$ for the flow (introduced in Remark \ref{remark}) as we have an amount of mass $q\geq 0$ flowing from a site containing the amount of mass $\zeta$ to a nearest one with a rate given by $Q(\zeta, q)=\frac{g(\zeta_1-q)}{g(\zeta_1)}S(q)$. In this way the reversibility condition is always satisfied for any choice of the function $g$ in \eqref{mis}. Special cases of the zero range models \eqref{zrrates} have been considered in \cite{FG}.

A special case is when the rate $R$ does not depend on the configuration $\zeta$, i.e., $R(\zeta,j)=R(j)$. This means that the random flow across the bonds happens always with the same distribution and the interaction is given just by the positivity constraint of the mass on each vertex. In this case, in order to have the gradient condition \eqref{grad} we need to impose that $R(j)=R(-j)$ so that $H(\zeta)=\int_0^{\zeta} R(j)jdj$ is symmetric as well and we obtain indeed a special case of the zero range family. We singled out this example because it is an extremely natural example when discussing the rates in terms of the current. The reversibility condition \eqref{fissoR} is obtained selecting $g\equiv 1$ and considering
$\psi(\zeta,j)=\psi(j)$, $j\geq 0$.
\end{example}
We stress that a zero range dynamics with arbitrary mass flowing is reversible only under condition \eqref{zrrates}. This is in contrast with the classic zero range models with particles and jumps of only one single particle, which is instead always reversible. In a forthcoming example we will discuss in detail the different behavior arising when one particle versus arbitrary amount of mass can jump.

\subsection{KMP-type models}\label{sec:kmp}

\begin{example}[KMP model]\label{exkmp}
This example is a special case of a bigger class, but we discuss it separately since it is a classic model.
The classic KMP (after Kipnis, Marchioro and Presutti) model \cite{KMP} corresponds to
$R(\zeta,j)=\frac{1}{\zeta_1+\zeta_2}$, i.e., the rate for observing a given current is uniform over all possible currents and such that the total rate is one, see also \cite{DCG}. In this case we obtain $H(a)=\frac a2$ and the reversibility condition is satisfied with the measures $g_\lambda$ of exponential form that correspond to the choice $g\equiv 1$ and $\lambda <0$. This uniform characterization of the model is a different viewpoint with respect to the classic definition in terms of redistribution of energy.
\end{example}

\begin{example}[Exponential invariant measures]\label{exco}
We consider here the exponential case that corresponds to $g=1$ and describe a class of models that are reversible and gradient. The KMP model of Example \ref{exkmp} is a special case. We consider the case $H'(a)\geq 0$, $a\geq 0$ that corresponds to a stability condition; when this condition is satisfied the mass is flowing on average from the site containing more mass toward the one containing less mass. Consider the function $h(a)=H'(a)>0$ and let $A(\zeta_1+\zeta_2, a)$ be, for any $\zeta_1,\zeta_2$, an antisymmetric function on the variable $a$. Then, we have that
\begin{equation}\label{evviva}
R(\zeta,j)=\left\{
\begin{array}{ll}
\frac{h(j)}{j}+\frac 1j A\left(\zeta_1+\zeta_2,j-\frac{\zeta_1+\zeta_2}{2}\right)\,, & j\geq 0\,,\\
R(\zeta,-j)\,, & j<0\,,
\end{array}
\right.
\end{equation}
solves \eqref{condpsi} with $H(t)=\int_{0}^t h(s)ds$ and $g\equiv 1$. The result follows by a direct computation observing that
$\int_0^bA(\zeta_1,\zeta_2, a)da$ is symmetric in $b$ for any $\zeta_1,\zeta_2$. Here again in \eqref{evviva} we need to restrict the choice of the functions respecting the positivity of the rates (here $h\geq 0$ is useful) and to take care of possible singularities.
Note that the classic KMP model is obtained with $h(j)=\frac 12$ and $A(s,y)=\frac ys$ getting
$$
R(\zeta,j)=R(\zeta,-j)=\frac{1}{2j} +\frac{1}{j(\zeta_1+\zeta_2)}\left(j-\frac{\zeta_1+\zeta_2}{2}\right)=\frac{1}{\zeta_1+\zeta_2}\,, \qquad j\geq 0\,.
$$
\end{example}

\begin{example}[Geometric invariant measures]\label{ex1}

We consider the discrete case in which the measure $g_\lambda(d\eta)$, defined in \eqref{mis}, is purely atomic with mass on the integer set $\mathbb N\cup\{0\}$ and $g\equiv 1$ so that the measures $g_\lambda$ are geometric. This is a discrete version of the previous class of examples. We describe a class of discrete models that are reversible with respect to this measure and gradient. In the discrete case the condition \eqref{condpsi} which guarantees the gradient property and reversibility with respect to the measure $g_\lambda(d\eta)$ becomes
\begin{align}\label{pot}
\sum_{j=0}^{\zeta_1}j\psi\pare{\zeta, j}-\sum_{j=0}^{\zeta_2}j\psi\pare{\zeta_1+j, \zeta_2-j, j}=g(\zeta_1)g(\zeta_2)\pare{H(\zeta_1)-H(\zeta_2)},
\end{align}
with $\psi(\zeta, j)=g(\zeta_1)g(\zeta_2)R(\zeta, j)$ for $j>0$.

In analogy with Example~\ref{exco} we consider the function $h(a)=H'(a)\geq 0$, $a\geq 0$ and let $A(\zeta_1+\zeta_2, a)$ be, for any $\zeta_1,\zeta_2\in \mathbb N\cup \{0\}$, an antisymmetric function on the variable $a$. The rate
\begin{equation}\label{rataj}
R(\zeta,j)=\frac{h(j)}{j}+\frac 1j A\left(\zeta_1+\zeta_2,j-\frac{\zeta_1+\zeta_2+1}{2}\right)\,, \qquad j> 0\,,
\end{equation}
satisfies \eqref{pot} with $H(k)=\sum_{j=0}^{k} h(j)$ and $g\equiv 1$. Note that with respect to the continuous case we had to shift the center of symmetry of the function $A$. Indeed,  using that $A(\zeta, a)$ is antisymmetric in the second variable, we get that \eqref{pot} is equivalent to:
\begin{equation*}
\sum_{j=0}^{\zeta_1}h(j)-\sum_{j=0}^{\zeta_2}h(j)
+\sum_{j=\zeta_1\wedge \zeta_2+1}^{\zeta_1\vee \zeta_2}A\pare{\zeta_1+\zeta_2, j- \frac{\zeta_1+\zeta_2+1}{2}}=H(\zeta_1)-H(\zeta_2)\,,
\end{equation*}
where the third term in the left hand side vanishes.
\end{example}

\subsection{Misanthrope process}\label{sec:mis}

We consider a natural variation of the zero range models, that has as special cases (in the discrete setting discussed later on) the exclusion process, the inclusion models and the misanthrope process \cite{GRV,KR,CC}. Models of this type have been considered by several authors as for example \cite{SS,FGS}. We investigate both reversibility
and the gradient property. We distinguish the continuous case from the discrete one and start discussing the continuous.

\begin{example}[Generalized misanthrope processes]\label{inclusion}
We consider a class of models for which the rate $Q(\zeta_1,\zeta_2,q)$ of observing a flow $q\geq 0$ is factorized in the form
\begin{equation}\label{rateinclusion}
Q(\zeta_1,\zeta_2, q)=b(\zeta_1,\zeta_2)S(q)\,,
\end{equation}
for arbitrary functions $b,S$ which we assume strictly bigger than zero in order to avoid degeneracies. This means that the rate is factorized and depends on the mass present not only at the starting site but also in the final one. Apart from the fact that the mass is continuous and an arbitrary amount $q$ may flow, this is the general form of the rates investigates in \cite{CC}.

We start by considering the reversibility.

\begin{lemma}\label{inclgrad}
A model with rates given in \eqref{rateinclusion} is reversible with respect to a product measure with marginals \eqref{mis} if and only if $b(\zeta_1,\zeta_2)=\frac{C(\zeta_1+\zeta_2)}{g(\zeta_1)g(\zeta_2)}$ for an arbitrary positive function $C:\mathbb R_+\to\mathbb R_+$.
\end{lemma}

\begin{proof}
	Since all rates are strictly positive, the detailed balance condition implies that
	\begin{equation}\label{ula}
	g(\zeta_1)g(\zeta_2)b(\zeta_1,\zeta_2)=g(\zeta_2+q)g(\zeta_1-q)b(\zeta_2+q,\zeta_1-q)\,, \quad \forall \zeta_1,\zeta_2,\,q\in[0,\zeta_1]\,.
	\end{equation}
	From \eqref{ula} we will deduce that
	\begin{equation}\label{unz}
	g(\zeta_1)g(\zeta_2)b(\zeta_1,\zeta_2)=g(\zeta_1')g(\zeta_2')b(\zeta_1',\zeta_2')\,,
	\end{equation}
	for any $(\zeta_1,\zeta_2), (\zeta_1',\zeta_2')$ such that $\zeta_1+\zeta_2=\zeta_1'+\zeta_2'$,
	and this implies directly
	$g(\zeta_1)g(\zeta_2)b(\zeta_1,\zeta_2)=C(\zeta_1+\zeta_2)$ for a suitable positive function $C$.
It remains to prove \eqref{unz}. In the case that $(\zeta_1',\zeta_2')=(\zeta_2+q,\zeta_1-q)$ for a $q\in[0,\zeta_1]$, \eqref{unz} follows directly from \eqref{ula}. If this is not the case we have to apply \eqref{ula} twice. Let $0<\epsilon<\min\{\zeta_1,\zeta_1'\}$. We apply \eqref{ula} with $q=\zeta_1-\epsilon$
obtaining $g(\zeta_1)g(\zeta_2)b(\zeta_1,\zeta_2)=g(\zeta_2+\zeta_1-\epsilon)g(\epsilon)b(\zeta_2+\zeta_1-\epsilon,\epsilon)$. We apply now \eqref{ula} once again with $\zeta_2+\zeta_1-\epsilon>q=\zeta_1'-\epsilon>0$ obtaining
$g(\zeta_2+\zeta_1-\epsilon)g(\epsilon)b(\zeta_2+\zeta_1-\epsilon,\epsilon)=g(\zeta_1')g(\zeta_2')b(\zeta_1',\zeta_2')$, which finishes the proof.
\end{proof}

Given a positive function $S$ we define
\begin{equation}\label{fm}
\hat S(\zeta):=\int_0^\zeta  j S(j) dj.
\end{equation}
The gradient condition in this case corresponds to having a function $H$ such that
\begin{equation}\label{sbaglirip}
b(\zeta_1,\zeta_2)\hat S(\zeta_1)-b(\zeta_2,\zeta_1)\hat S(\zeta_2)=H(\zeta_1)-H(\zeta_2)\,,
\end{equation}
which is a condition on the form of the antisymmetric part of the function $b\hat S$. The model is therefore gradient if and only if there exists a symmetric function $\mathbb S$ and a function $H$ such that
\begin{equation}\label{simmgra}
b(\zeta_1,\zeta_2)\hat S(\zeta_1)=\mathbb S(\zeta_1,\zeta_2)+\frac 12\left(H(\zeta_1)-H(\zeta_2)\right)\,.
\end{equation}
A less abstract characterization of the gradient condition is difficult in the general case; we consider a special case and discuss a few more cases in Example~\ref{specialI}.

\begin{lemma}\label{incl}
A model with rates given in \eqref{rateinclusion} with $b(\zeta_1,\zeta_2)=\ell(\zeta_1)L(\zeta_2)$, for smooth and positive functions $\ell, L$, is gradient if and only if one of the following two conditions are satisfied
\begin{equation}\label{piuk}
\left\{
\begin{array}{l}
L(\zeta)=c\ell(\zeta)\hat S(\zeta)+k\,,\\
\ell(\zeta)=c/\hat S(\zeta)\,,\\
\end{array}
\right.
\end{equation}
for arbitrary constants $c,k$ and $\hat S$ given in \eqref{fm}.
\end{lemma}
\begin{proof}
The gradient condition  \eqref{grad} is satisfied if there exists a real function $H$ such that \eqref{sbaglirip} is satisfied with $b(\zeta_1,\zeta_2)=\ell(\zeta_1)L(\zeta_2)$. Calling $\hat\ell:=\ell \hat S$
and taking the derivatives $\frac{\partial^2}{\partial\zeta_1\partial\zeta_2}$ on both sides of \eqref{sbaglirip} we obtain
that we need to have $\hat{\ell}'(\zeta_1)L'(\zeta_2)=\hat{\ell}'(\zeta_2)L'(\zeta_1)$. This relation can be satisfied either if $\hat{\ell}'=0$, which gives the second condition in \eqref{piuk} or if $L'/\hat \ell '$ is constant, which is the first condition in \eqref{piuk}. Hence we proved that \eqref{piuk} are necessary conditions; in order to see that they are also sufficient it is enough to insert the special forms in \eqref{sbaglirip}.
\end{proof}

If we combine the results of lemmas \ref{inclgrad} and \ref{incl} we obtain models that are at the same time gradient and reversible. This is possible when $C(\zeta_1+\zeta_2)$ is constant and $\ell$ and $L$ are proportional. We have two possibilities; when the first condition in \eqref{piuk} is satisfied then we have
that the rates are given by
\begin{equation}\label{RGPR0}
Q(\zeta_1,\zeta_2, q)= \frac{k S(q)}{\left[1-c\hat S(\zeta_1)\right]\left[1-c\hat S(\zeta_2)\right]}\,,
\end{equation}
for suitable constants $c,k$ and a function $S$;
the corresponding $g$ is $g(\zeta)=1-cC\hat S(\zeta)$.
If instead the second condition in \eqref{piuk} is satisfied  we have
that the rates are given by
\begin{equation}\label{RGPR}
Q(\zeta_1,\zeta_2, q)=\frac{k S(q)}{\hat S(\zeta_1) \hat S(\zeta_2)}\,,
\end{equation}
for a suitable constant $k$ and a function $S$; the corresponding $g$ is given by $g(\zeta)=\hat S(\zeta)$.

\end{example}

\begin{example}[Discrete generalized misanthrope processes]\label{inclusiond}
We consider here the discrete version of the previous example which will give as special cases the exclusion and the inclusion processes. Mass is assuming only integer values and we do not discuss the cases when the number of particles on each site is bounded. Some computations are similar to the previous case but in the discrete setting it is simpler and natural to consider constraints on the amount of mass flowing. More precisely the rates that we consider are again of the form \eqref{rateinclusion} with $\zeta_1,\zeta_2, q$ that assume only integer values. We will however distinguish the following cases:
\begin{enumerate}
\item $S(q)>0$ for all $q$  (case (A)),
\item $S(1)=1$ and $S(q)=0$ for any $q\geq 2$ (case (B)).
\end{enumerate}
Note that the case (B) corresponds to the situation when only one particle can jump in a single step (as in the classic models). The rates of the case (B) correspond to the special models introduced in \cite{CC} under the name of \emph{Misanthropic processes}. In particular in \cite{CC} a collection of discrete equations characterizing reversibility have been introduced. We will give a complete solution to such conditions extending the result in \cite{CC} which, to the best of our knowledge,
is new.

Here, again, for the gradient condition we are not aiming to a complete characterization but we restrict to the same special form of Lemma \ref{incl}. As in the continuous case we set $\hat S(\zeta):=\sum_{j=1}^\zeta j S(j)$.
\begin{lemma}
Consider discrete models with rates as in \eqref{rateinclusion} with $b(\zeta_1,\zeta_2)=\ell(\zeta_1)L(\zeta_2)$ and $\zeta_1, \zeta_2$ assuming integer values only.
In the case (A) a model is gradient if and only if the following condition holds
\begin{equation}\label{piuk2}
L(\zeta)=c\ell(\zeta)\hat S(\zeta)+k\,,
\end{equation}
for arbitrary constants $c,k$. In the case (B) a model is gradient if and only if the following condition holds
\begin{equation}\label{piuk2B}
L(\zeta)=c\ell(\zeta)+k\,,
\end{equation}
for arbitrary constants $c,k$.
\end{lemma}

\begin{proof}
Consider first the case $(A)$. Condition \eqref{piuk2} is clearly sufficient. We show that it is necessary too.
Similarly to the continuous case the gradient condition is satisfied if there exists a real function $H$ that satisfies \eqref{sbaglirip}.
Let us call $\hat \ell(\zeta):=\ell(\zeta)\hat S(\zeta)$ and consider $\zeta_2=0$. Recall that in the discrete case we need to have $\ell(0)=0$ and moreover  we fix the arbitrary additive constant of the function $H$ in such a way that $H(0)=0$. We obtain that $H(\zeta_1)=\hat\ell(\zeta_1)L(0)$.
Replacing back to \eqref{sbaglirip} and computing the equality for a fixed $\zeta_2^*$ for which $\hat \ell(\zeta_2^*)>0$ (that must exist otherwise the model is trivial) we obtain
\begin{equation*}
\hat \ell(\zeta_2^*)L(\zeta_1)=(L(\zeta_2^*)-L(0))\hat \ell(\zeta_1)+L(0)\hat\ell(\zeta_2^*)\,,
\end{equation*}
that is equivalent to \eqref{piuk2}.

\smallskip
For the case $(B)$ the gradient condition becomes
\begin{equation*}
\ell(\zeta_1)L(\zeta_2)-\ell(\zeta_2)L(\zeta_1)=H(\zeta_1)-H(\zeta_2)\,,
\end{equation*}
and condition \eqref{piuk2B} is obtained by the same argument as before with $\hat \ell$ replaced by $\ell$.
\end{proof}
Note that \eqref{piuk2B} coincides with \eqref{piuk2} since in the case (B) we have $\hat S(\zeta)=\id(\zeta >0)$. Note also that the second condition in \eqref{piuk} in the continuous case, cannot be satisfied in the discrete setting, since we have $\hat S(0)=0$ and $\ell(0)=0$ and the constant $c$ would be zero giving a trivial model.

We consider now reversibility.
\begin{lemma}\label{a+b}
In the case $(A)$ we have reversibility  with respect to a product measure \eqref{mis} if and only if $b(\zeta_1,\zeta_2)=\frac{C(\zeta_1+\zeta_2)}{g(\zeta_1)g(\zeta_2)}$ for an arbitrary positive function $C$. In the case $(B)$ we have reversibility if an only if
\begin{equation}\label{cocoz}
b(\zeta_1,\zeta_2)=B(\zeta_1)\phi(\zeta_1,\zeta_2),
\end{equation}
where $B$ is an arbitrary positive function and $\phi$ is a function that satisfies the symmetry
\begin{equation}\label{simcoc}
\phi(l,m)=\phi(m+1,l-1)\,, \qquad l,m \in \mathbb N\cup\{0\}; l\geq 1\,.
\end{equation}
In this case we have that the product invariant measure is characterized by the weights
\begin{equation}\label{ciaoroger}
g(k)=\left\{
\begin{array}{ll}
1 & k=0\,,\\
\prod_{j=1}^k B(j)^{-1} & k\geq 1\,.
\end{array}
\right.
\end{equation}
\end{lemma}

\begin{proof}
The proof of the case (A) is obtained with a discrete version of the argument in the continuous case.

We consider now the case (B). We first check that if \eqref{cocoz}, \eqref{simcoc} and \eqref{ciaoroger} are satisfied then the detailed balance holds. We have to check that for any $\zeta_1\geq 1$ and $\zeta_2 \geq 0$ we have
$$
g(\zeta_1)g(\zeta_2)B(\zeta_1)\phi(\zeta_1,\zeta_2)=g(\zeta_2+1)g(\zeta_1-1)B(\zeta_2+1)\phi(\zeta_2+1,\zeta_1-1)\,.
$$
Due to \eqref{simcoc} the factors $\phi$ simplify and the equality follows since $\frac{g(\zeta+1)}{g(\zeta)}=\frac{1}{B(\zeta+1)}$.

We verify now that \eqref{cocoz}, \eqref{simcoc} and \eqref{ciaoroger} are also necessary conditions to have detailed balance, i.e., for the validity of
$$
g(\zeta_1)g(\zeta_2)b(\zeta_1,\zeta_2)=g(\zeta_2+1)g(\zeta_1-1)b(\zeta_2+1,\zeta_1-1)\,,
$$
for a suitable positive function $g$. We introduce the function $B(l):=\frac{g(l-1)}{g(l)}$ for $l\geq 1$ and $B(0)=0$. Using this function the detailed balance becomes
$$
\frac{b(\zeta_1,\zeta_2)}{B(\zeta_1)}=\frac{b(\zeta_2+1,\zeta_1-1)}{B(\zeta_2+1)}\,,
$$
which is exactly \eqref{simcoc} with $\phi=\frac{b}{B}$.

\end{proof}

\begin{remark}
From the above analysis it follows that the whole family of reversible models of case (B) is parametrized by the arbitrary positive function $B$
and the function $\psi$ satisfying \eqref{simcoc}. It is rather simple to construct a function satisfying \eqref{simcoc}; referring to Figure \ref{lafigura} we can fix arbitrarily the values of the function on the
points with integer coordinates that are marked with a round black dots, the values on the points marked with black squares are then fixed by the value on the paired point with a black dot. This is because the symmetry in \eqref{simcoc} is an involution.
\end{remark}
 \begin{figure}\label{lafigura}
	\centering
	\includegraphics{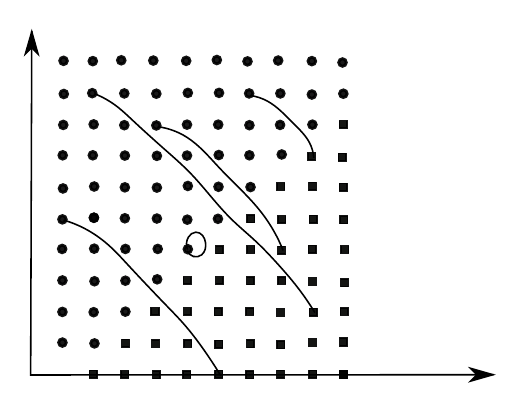}
	\caption{Integer points are marked either with a circle or with a square. On the points with a circle the value of the function $\psi$ can be defined arbitrarily; on the points with a square the values should be set equal to the corresponding point with a circle. We draw some lines among points that are in correspondence. Points that are exactly on the line just below the main diagonal are mapped to themself and the function $\psi$ can be defined arbitrarily there.}
\end{figure}
\begin{remark}
A particularly interesting example of reversible rates of case (B) is given by $b(\zeta_1,\zeta_2)=\ell(\zeta_1)L(\zeta_2)$ where $\ell, L$ are two arbitrary positive measures (with $\ell(0)=0$). In that case we have $B(l)=\frac{\ell(l)}{L(l-1)}$ and $\phi(l,m)=L(l-1)L(m)$.
Remarkable examples of this form are the exclusion process that corresponds to $\ell(0)=1-\ell(1)=0$
and $L(1)=1-L(0)=0$; the inclusion process that corresponds to $\ell(\zeta)=\zeta$ and $L(\zeta)=m+\zeta$, and finally when $L=1$ we have the classic zero range dynamics.

Another special case discussed in \cite{CC} is given by $b(\zeta_1,\zeta_2)=L(\zeta_1)L(\zeta_2)C(\zeta_1+\zeta_2)$ for arbitrary positive functions $L,C$ (with $L(0)=0$). In this case we have indeed $B(l)=\frac{L(l)}{L(l-1)}$ and $\phi(l,m)=L(l-1)L(m)C(l+m)$ and the validity of \eqref{simcoc} is easily verified.

Finally we consider some general rates of the form $\ell(l)L(m)\tilde\phi(l,m)$ with $\ell, L$ being arbitrary positive functions (again $\ell(0)=0$) and $\tilde \phi$ is a function satisfying \eqref{simcoc}. Every rate of this form is reversible but this fact it is not enlarging the class individuated by Lemma \ref{a+b} since every rate of this form can be written as \eqref{cocoz} with $\psi$ satisfying \eqref{simcoc} by considering $B(l)=\frac{\ell(l)}{\ell(l-1)}$ and $\phi(l,m)=\ell(l-1)L(m)\tilde\phi(l,m)$.
\end{remark}

\begin{remark}
In \cite{CC} the author considers models of case (B) with in addition a spatial structure. In the case of nearest-neighbor jumps with symmetric spatial rates, the author's condition for the reversibility with respect to a product measure (which is necessary and sufficient in dimension one) is given by
\begin{equation}\label{equanozza}
\frac{b(i,j)}{b(j+1,i-1)}=\frac{b(i,0)b(1,j)}{b(j+1,0)b(1,i-1)}\,, i\geq 1, j\geq 0\,.
\end{equation}
It is easy to check that our general solution \eqref{cocoz}, \eqref{simcoc}, \eqref{ciaoroger} is the general solution to \eqref{equanozza}. This was not obtained in \cite{CC} and we could not find this fact in more recent references either.

It is interesting to observe that in the case of non-symmetry in space of the rate of jump, the condition that has to be added in \cite{CC} to \eqref{equanozza} is
\begin{equation}
b(i,j)-b(j,i)=b(i,0)-b(j,0)\,,
\end{equation}
that is indeed equivalent to the gradient condition for the models of case (B). This is a special case of a general fact that gradient and reversible models preserve their invariant measure also in presence of an external field that creates an asymmetry on the rates (see \cite{BDGJLstat} section 2.5).
\end{remark}

If we want to construct discrete models that are at the same time gradient and reversible then, restricting to the case when $b=\ell L$, in the case $(A)$ we have exactly the same formulas of the continuous models \eqref{RGPR0}, while in the case $(B)$ we obtain the whole class of models with the only constraints that $\ell$ and $L$ should be given by \eqref{piuk2B} and the function $g$ by \eqref{ciaoroger}.

\end{example}

\begin{example}[Special identities]\label{specialI}
	Another simple and natural class of models is obtained by selecting the function $\psi$ (introduced in Lemma \ref{Rtilde}) in the form
	$\psi(\zeta,j)=g(\zeta_1)g(\zeta_2)M(\zeta_1+\zeta_2)S(j)$ (note that is has the same form with the misanthrope model). In the special case $g=1$, condition \eqref{condpsi} is satisfied finding a triple of functions such that
	\begin{equation}\label{identita}
	M(\zeta_1+\zeta_2)\Big(\hat S(\zeta_1)-\hat S(\zeta_2)\Big)=H(\zeta_1)-H(\zeta_2)\,,
	\end{equation}
	where $\hat S(y)=\int_0^yS(x)xdx$.
	The KMP model is of this type choosing $M(s)=1/s$ and $S(j)=2j$, in which case the identity \eqref{identita} becomes
	$$
	\frac{\zeta_1^2-\zeta_2^2}{\zeta_1+\zeta_2}=\zeta_1-\zeta_2\,.
	$$
	We can generate other gradient and reversible models finding functions satisfying relation \eqref{identita}. Observing that
	$$
	\left(e^{-\zeta_2}-e^{-\zeta_1}\right)e^{\zeta_1+\zeta_2}=e^{\zeta_1}-e^{\zeta_2},
	$$
	we deduce that we obtain such a model choosing $S(j)=\frac{e^{-j}}{j}$, $M(s)=e^s$ and $H(x)=e^x$. The corresponding rates are therefore
	\begin{equation}\label{studpap}
	R(\zeta, j)= R(\zeta, -j)= e^{\zeta_1+\zeta_2}\frac{e^{-j}}{j}\,, \qquad j\geq 0\,,
	\end{equation}
	which indeed is also a special case of the models discussed in Example \ref{inclusion}. Hence, we just showed how to reach a similar characterization by alternative paths.
\end{example}

\begin{remark}
Observe that the example with rates \eqref{studpap} gives $\int_{-\zeta_2}^{\zeta_1}R(\zeta,j)dj=+\infty$. This happens also in models  like in formula \eqref{evviva} for example with $h=1$ and $A=0$. Models of this form have an infinite rate of jump across each edge and have therefore a dense set of jump times. The process is well defined since with very high rate we observe very small currents.  We are in presence of a class of models that can be imagined as interacting Levy processes of pure jump type. A fully mathematical discussion of such models is beyond the aim of this paper.
See also \cite{FFG} where models of this type were considered simultaneously to us.
\end{remark}

\section{Transport coefficients}\label{sec:TC}

Here we give a short derivation of the connection of the transport coefficients to the details of the microscopic dynamics. The fluctuations at the level of large deviations for diffusive models are determined by the transport coefficients alone, i.e., the diffusion coefficient and the mobility \cite{mft}. We assume that the dynamics is gradient \eqref{grcond-giu} and reversible with respect to a product measure with marginals $g_\lambda$. We consider in this section the lattice $\Lambda_N=\frac 1N\left(\mathbb Z^d/(N\mathbb Z)^d\right)$, i.e. the $d$-dimensional discrete torus with edges among the nearest neighbor lattices sites
and mesh equal to $1/N$. We consider the discrete torus $\Lambda_N$ embedded in the continuous torus $\Lambda:=[0,1]^d$ with periodic boundary conditions. The analytic form of the transport coefficients can be deduced from the diffusive hydrodynamic scaling limit. A fully mathematical proof of the scaling limit and the corresponding large deviations for our models is a difficult and technical task and we give here a heuristic argument that gives the analytic formulas for the diffusion matrix and for the mobility. Since we consider isotropic models mobility and diffusion matrix will be multiples of the identity matrix and we can just consider the proportionality factor that we call diffusion coefficient and mobility.

\subsection{Hydrodynamic limit}\label{HL}

We consider the class of models above discussed in presence of a weak external field.
When we switch on a weak external field we have that the rates of jump are perturbed. Let $F$ be a smooth vector field on $\Lambda$ and let $F_N$ be its discretized version defined as
\begin{equation}\label{discvf}
F_N(x,y):=\int_{(x,y)}F\cdot dl\,,\qquad (x,y)\in E_N\,,
\end{equation}
that is the line integral along the oriented segment $(x,y)$. Note that $F_N(x,y)=O(1/N)$, since our lattice has mesh $1/N$, and this fact is the motivation on the name weak external field. The perturbed rates are given by
\begin{equation}\label{RF}
R^{F}(\eta(x),\eta(y),j)=R(\eta(x),\eta(y),j)e^{F_N(x,y)j}
\end{equation}
and consequently the instantaneous current is given by
\begin{eqnarray}\label{instcur}
j_\eta^F(x,y)&=&\int_{-\eta(y)}^{\eta(x)}R(\eta(x),\eta(y),j)e^{F_N(x,y)j}jdj \\
&=&j_\eta(x,y)+F_N(x,y)\int_{-\eta(y)}^{\eta(x)}R(\eta(x),\eta(y),j)j^2dj+O(1/N^2)\,.
\end{eqnarray}
The motivation behind the choice \eqref{RF} of the rates is that the perturbed part is related to the work of the vector field $F_N$ with a flow of mass $j$.

We consider the diffusive scaling limit of the process under a weakly asymmetric external field $F$. The diffusive rescaling consists in multiplying the rates by $N^2$ and in multiplying the mass on each site by a factor of $N^{-d}$; consequently the instantaneous current is multiplied by a factor of $N^{2-d}$.

The empirical measure is an element of $\mathcal M^+(\Lambda)$, the positive measures on $\Lambda$, associated to a configuration $\zeta$ and defined by
\begin{equation}
\pi_N(\zeta):=\frac{1}{N^d}\sum_{x\in \Lambda_N}\zeta(x)\delta_x\,,
\end{equation}
where $\delta_x$ denotes the delta measure at $x\in \Lambda$.
Consider an initial state $\eta$ associated to a given measurable profile $\rho_0(x)$, i.e., such that
$\lim_{N\to +\infty} \pi_N(\eta)= \rho_0(x)dx$,
where the convergence is in the weak sense of measures, i.e.,
for any continuous function $f$ we have
\begin{equation}\label{initial}
\lim_{N\to +\infty} \int_{\Lambda} f d\pi_N(\eta)=\lim_{N\to +\infty}\frac{1}{N^d}\sum_{x\in \Lambda_N}\eta(x)f(x)= \int_\Lambda\rho_0(x)f(x)dx\,.
\end{equation}
In order to determine the hydrodynamic behavior of the current under this rescaling, we can consider the scalar product of the current observed on the lattice with a discretized vector field. Consider $G$ a smooth vector field and let $G_N$ be its discretized version as defined by the
definition in \eqref{discvf}.

Recall the definition \eqref{defJ} of the net current across the edges that define a discrete vector field on the lattice for each time $t$ and call $J^F$ the current for the accelerated process associated to the external field $F$. The scalar product of the net current $J^F$ with the discretized vector field $G_N$ is given by
\begin{equation}\label{original}
\frac{1}{2N^d}\sum_{(x,y)\in E_N}J^F(x,y,t)G_N(x,y)\,.
\end{equation}
We argue now that \eqref{original} is converging for large $N$ to $\int_0^tds \int_\Lambda dx\, J(\rho(x,s))\cdot G(x,s)$, where $J(\rho)$ is the typical current in the scaling limit associated to the density profile $\rho$.
By \eqref{defmart} (with $j^F$ given in \eqref{instcur}), up to a negligible martingale, formula \eqref{original} coincides with
\begin{equation}\label{uno}
\frac{N^{2-d}}{2}\sum_{(x,y)\in E_N}\int_0^tj_{\eta_s}^F(x,y)G_N(x,y)ds\,,
\end{equation}
where we recall that the power of the $N$ factor in front is due to the acceleration by $N^2$ of the rates.
The $1/2$ factor in \eqref{original} and \eqref{uno} is due to the fact that we are summing twice every unoriented edge.
Under general assumptions that we do not discuss, the martingale term is negligible by Doob's inequality.
Using the gradient condition \eqref{grcond-giu} and some discrete integrations by parts, formula \eqref{uno} can be written as
\begin{align}\label{al}
& \frac{1}{N^d}\sum_{x\in \Lambda_N}\int_0^t\tau_x h(\eta_s)N^2\nabla \cdot G_N(x)ds\nonumber \\
&+\frac{1}{2N^d}\sum_{(x,y)\in E_N}\int_0^tNF_N(x,y)NG_N(x,y)M_2(\eta_s(x),\eta_s(y))ds+ O(1/N)\,.
\end{align}
In the above formula we used the discrete divergence of the vector field $G_N$ defined by
\begin{equation}
\nabla\cdot G_N(x):=\sum_{(x,y)\in E_N}G_N(x,y)\,,
\end{equation}
for which it is possible to check that $N^2\nabla\cdot G_N(x)=\nabla \cdot G(x)+O(1/N)$ where the infinitesimal term is uniform in $x$ ($\nabla \cdot$ is the discrete divergence on the left hand side and the continuous one on the right hand side). We also used the notation
\begin{equation}
M_2(\eta(x),\eta(y)):= \int_{-\eta(y)}^{\eta(x)}R(\eta(x),\eta(y),j)j^2dj\,,
\end{equation}
for the second moment of the measure $R$. We remark that the terms $NF_N(x,x+e_i)$ and $NG_N(x,x+e_i)$, where $e_i$ is the vector with all coordinates equal to zero apart coordinate $i$ that is equal to $1/N$, are equal to $F_i(x)$ and $G_i(x)$ (which are the $i$-th coordinates of the smooth vector fields $F$ and $G$) up to a uniform infinitesimal term.

\smallskip

We define
\begin{equation}\label{trc}
\left\{
\begin{array}{l}
\hat D(\rho):=\mathbb E_{g_{\lambda(\rho)}}(h(\eta))\,,\\
\sigma(\rho):=\frac{1}{2}\mathbb E_{g_{\lambda(\rho)}\otimes g_{\lambda(\rho)} }(M_2(\eta(x),\eta(y)))\,,
\end{array}
\right.
\end{equation}
where $\lambda(\rho)$ is the inverse function of \eqref{robeta} and $g_\lambda\otimes g_\lambda$ is the product measure of two copies of $g_\lambda$. The factor $1/2$ in the second formula is present just to be consistent with the classic formulas.

The scaling limit for diffusive models is based on a replacement lemma (see \cite{KL, Spohn} for technical details) whose proof for our class of models may not be trivial. We just derive heuristically the scaling limit. The informal statement of the replacement lemma is as follows. Let $a:\Lambda_N\to \mathbb R$ be a local function, and let $A(\rho):= \mathbb E_{\otimes g_{\lambda(\rho)}}(a(\eta))$ be the average with respect to a product measure with marginals $g_{\lambda(\rho)}$. Consider an initial condition such that \eqref{initial} is satisfied. We have that for any $\epsilon >0$ and any smooth test function $\phi:\Lambda \to \mathbb R$ there exists a density $\rho_t(x)$ such that
\begin{equation}\label{replacement}
\mathbb P\left(\left|\frac{1}{N^d}\sum_{x\in \Lambda_N}\tau_xa(\eta_t)\phi(x)-\int_{\Lambda}A(\rho_t(x))\phi(x)dx\right|>\epsilon\right)
\end{equation}
is converging to zero when $N\to +\infty$. Applying \eqref{replacement} to formula \eqref{al} we obtain that the random expression is converging with high probability to the deterministic limit
\begin{equation}\label{diventabestia}
\int_0^tds\,\int_{\Lambda} \left[\hat D(\rho_s(x))\nabla \cdot G(x)+2\sigma(\rho_s(x)) F(x)\cdot  G(x)\right] dx\,,
\end{equation}
for any smooth test vector field $G$.
After an integration by parts, formula \eqref{diventabestia} coincides with
\begin{equation}\label{utile}
\int_0^tds\int_{\Lambda} \left[ -D(\rho_s(x))\nabla\rho_s(x) +2\sigma(\rho_s(x)) F(x)\right] \cdot  G(x)dx=\int_0^tds\int_{\Lambda}J(\rho_s(x))\cdot G(x)dx,
\end{equation}
where
\begin{equation}
J(\rho_t(x)):=-D(\rho_t(x))\nabla\rho_t(x) +2\sigma(\rho_t(x)) F(x)
\end{equation}
is the typical current associated to the density $\rho$ and
\begin{equation}\label{D}
D(\rho):=\frac{d}{d\rho}\hat D(\rho).
\end{equation}
It remains to characterize the density profile $\rho_t$. Since our stochastic system conserves the mass its limiting equation is therefore a continuity equation. By \eqref{utile} we identified the typical current in terms of $\rho_t$. We deduce that the limiting hydrodynamic equation for the density $\rho_t$ is given by the following nonlinear diffusion equation
\begin{equation}
\partial_t\rho=-\nabla\cdot J(\rho)=\nabla \cdot\left[D(\rho)\nabla\rho -2\sigma(\rho) F\right]\,.
\end{equation}
We have therefore that the diffusion matrix is $D(\rho)\id$ and the mobility is $\sigma(\rho)\id$.

\subsection{Free energy}

The density of the free energy for a model with invariant measure given by a product measure with marginals $g_\lambda$ is related to the large deviations rate functional for the empirical measure. The pressure $\mathcal P_\lambda$ associated to the empirical measure when distributed according to the product distribution $\otimes_{x\in \Lambda_N} g_\lambda(d\eta(x))$ is defined as follows. Let $f:\Lambda\to \mathbb R$ be a continuous function, we have
\begin{align}
\mathcal P_\lambda(f)&:=\lim_{N\to +\infty} \frac {1}{N^d}\log \mathbb E_{\otimes g_\lambda}\left(e^{N^d\int d\pi_N(\eta)f}\right) \nonumber \\
&=\int_{\Lambda}\log\left(\frac{Z(\lambda +f(x))}{Z(\lambda)}\right)dx\,.
\end{align}
The above exact computation can be easily done since the measure is product.
The corresponding large deviations rate functional is the Legendre transform
$I_\lambda(\rho)=\sup_f\left\{\int fd\rho-\mathcal P_\lambda(f)\right\}$, which is $+\infty$ unless $\rho$ is an absolutely continuous element of $\mathcal M^+(\Lambda)$ and for absolutely continuous measures assumes the value
\begin{equation}\label{rfem}
I_\lambda(\rho)=\int_{\Lambda}dx \left\{\rho(x)\big[\lambda(\rho(x))-\lambda\big]-\log\frac{Z(\lambda(\rho(x)))}{Z(\lambda)}\right\}\,.
\end{equation}
In formula \eqref{rfem} we still denote by $\rho$ the density of the measure $\rho$ and $\lambda(\rho)$
is the inverse function of \eqref{robeta}.
The density of free energy is therefore
\begin{equation}
f(\rho)=\rho \lambda(\rho)-\log Z(\lambda(\rho))\,,
\end{equation}
and \eqref{rfem} can be therefore written in the classic equilibrium form
\begin{equation}
I_\lambda(\rho)=\int_{\Lambda}dx\big[f(\rho(x))-f(\rho(\lambda))-f'(\rho(\lambda))\big(\rho(x)-\rho(\lambda)\big)\big]\,,
\end{equation}
were we used that $f'(\rho)=\lambda(\rho)$, that can be deduce using \eqref{robeta}.
The second derivative is given by $f''(\rho)=\lambda'(\rho)$ and, for stochastic particle systems, it is involved in the Einstein relation which is given by
\begin{equation}\label{ein}
D(\rho)=\sigma(\rho)f''(\rho),
\end{equation}
where $D$ and $\sigma$ are the transport coefficients introduced in the previous section.
The factor $1/2$ in \eqref{trc} was introduced to avoid a factor here. A general proof of \eqref{ein} from formulas \eqref{trc} seems to be not trivial but we will show its validity for classes of models.

\subsection{Integrability}\label{sec:int}

A problem of interest, which is the central theme of this paper, is to study the fluctuations of diffusive systems in contact with boundary reservoirs. In this physical situation we have a stationary non equilibrium state characterized by the invariant measure of the system that in general has statistical properties very different from the states of equilibrium. A direct computation of the unique invariant measure in the non-equilibrium situation is very difficult even if in some cases very interesting and rich combinatorial constructions have been devised. In the special case when the direct approach is possible, typically one observes that in the stationary non-equilibrium case there are long range correlations that are in correspondence with a non local large deviations rate functional.
As a special case we have the classic zero range model. Differently from the other models we have that the invariant measure is always of product type and no long range correlations appear \cite{DMF,EH,mft}. This fact is true also for the Ginzburg Landau model \cite{GPV, Spohn}.

An indirect approach to the study of fluctuations for stationary non-equilibrium states is a dynamical one. This is the approach developed by the Macroscopic Fluctuation Theory \cite{mft}. According to this general theory the structure of the fluctuations of stationary non equilibrium states is determined just by the transport coefficients and some macroscopic details of the sources in contact with the boundary of the system. The fluctuations for the SNS are related to the so called \emph{Quasipotential} that is computed by a minimization over the dynamic large deviations rate functional under suitable constraints. As a consequence the quasipotential is a solution of a stationary infinite dimensional Hamilton-Jacobi equation.

There is a whole class of one dimensional models for which it is possible to obtain a closed explicit form of the large deviations rate functional for the invariant measure \cite{mft,BGL}.
The large deviations rate functional is in this case a non local functional depending on the solution of a nonlinear differential problem with boundary conditions. The form of the boundary conditions depends on the interaction with the boundary sources (see \cite{mft} for strong interaction but also more recently \cite{BL,DHS,LV} for weak interactions). To this class of models belong the exclusion process as well as the classic KMP process and its dual \cite{BGL}, but in general the solvability of the boundary driven rate functional is shared by any class of models having a constant diffusion coefficient and a second order polynomial as mobility \cite{BGL}, i.e., when we have
\begin{equation}\label{remkp2}
\left\{
\begin{array}{l}
D(\rho)=c\,,\\
\sigma(\rho)=\mathcal P_2(\rho)\,.
\end{array}
\right.
\end{equation}
The non local form of the rate functional in this case is given by formulas (7.9) and (7.10) in \cite{BGL}. It is important to observe that the form of the rate functional is explicit enough to deduce for example the correlation functions. See also \cite{G1,G2} for a detailed discussion of this with also some generalizations.

A second remarkable class of transport coefficients are those satisfying the relation
\begin{equation}\label{remk}
D(\rho)\sigma''(\rho)=D'(\rho)\sigma'(\rho)\,.
\end{equation}
Recalling \eqref{ein}, it can be proved (see \cite{mft}) that this condition is satisfied either  when $\sigma(\rho)$ is constant and then $D(\rho)>0$ is arbitrary (a model that has transport coefficients of this type is the  Ginzburg-Landau model \cite{GPV, Spohn}) or otherwise when $D(\rho)=c\sigma'(\rho)$ for an arbitrary constant $c>0$. A model that has the transport coefficients of this form with $c=1$ is the classic zero range model.

It is not easy to explicitly construct models satisfying such conditions since the transport coefficients can be computed only for diffusive models that are at the same time reversible and of gradient type.
Otherwise only variational representations of the transport coefficients are available.  We discussed however in the previous sections a large class of models that are at the same time gradient and reversible and it is therefore natural and interesting to investigate which are the models having such remarkable form of the transport coefficients.
We stress again that we did not develop a rigorous derivation of the form of the transport coefficients.

\section{Zero range dynamics}\label{sec:ZRD}

In this section we discuss in detail the zero range dynamics obtaining a complete characterization of integrability, in the sense of Subsection~\ref{sec:int}. Note that the class of zero range dynamics discussed here is much wider than the classic zero range model. Indeed, we will see that having a product invariant measure outside equilibrium is a feature that is in general lost.

First of all we observe that we can compute the transport coefficients using formulas \eqref{trc} and \eqref{zrrates}. We start with the mobility,
\begin{align}\label{mob}
\sigma(\rho)&=\frac 12\mathbb E_{g_{\lambda(\rho)}\otimes g_{\lambda(\rho)}}\left(\int_{-\zeta_2}^{\zeta_1}R(\zeta_1, \zeta_2,j)j^2dj\right)\\
&=\int_0^\infty d \zeta_1\int_0^\infty d \zeta_2 \int_0^{\zeta_1}d j\, \frac{g(\zeta_1-j)}{g(\zeta_1)}S(j)j^2\frac{e^{\lambda(\zeta_1+\zeta_2)}}{(Z(\lambda))^2}g(\zeta_1)g(\zeta_2)\\
&=\int_0^\infty d \zeta_1\int_0^{\zeta_1}d j\,\frac{g(\zeta_1-j)}{g(\zeta_1)}S(j)j^2\frac{e^{\lambda\zeta_1}}{Z(\lambda)}g(\zeta_1)\\
&=\int_0^\infty d \zeta_1\int_0^{\zeta_1}d jg(\zeta_1-j)S(j)j^2e^{\lambda j}\frac{e^{\lambda(\zeta_1-j)}}{Z(\lambda)}\\
&=\int_0^\infty S(j)j^2e^{\lambda j}dj\,.
\end{align}
Note that in the above computation we need to assume $\lambda\equiv\lambda(\rho)<\lambda_c$, since otherwise the invariant measure is not defined. We discover however that the formula for the mobility may give an infinite value for some subcritical values of $\lambda$. Indeed we have that, depending on the function $S$, only for $\lambda<\tilde\lambda_c\leq \lambda_c$ the mobility is finite. Since the function $S$ is determining the distribution of the amount of mass that is flowing in a single jump, we have that there may be a second critical value such that it is not expected a hydrodynamic behavior above this threshold. The value of $\tilde \lambda_c$ may coincide with $\lambda_c$ or even be equal to $-\infty$, in such a case for these models we do not expect a hydrodynamic behavior. This phenomenon is related to the fact that, due to the possibility of a large amount of mass flowing in a single step, there is no reason to expect an absolutely continuous measure in the scaling limit. This is a mechanism similar to the one is behind the possibility of condensation. A rigorous discussion of the scaling limit of such models is challenging and we do not consider this issue here.

The coefficient $\hat D(\rho)$ defined in \eqref{trc}, using the local function $h$ as defined in Example~\ref{zre} and Example~\ref{zre2}, is given by
\begin{align*}
\hat D(\rho)&=\mathbb E_{g_{\lambda(\rho)}}\left(\int_0^{\zeta}\frac{g(\zeta-j)}{g(\zeta)}S(j)jdj\right)\\
&=\int_0^\infty d \zeta \int_0^\zeta d j\,  g(\zeta-j)S(j)j\frac{e^{\lambda\zeta}}{Z(\lambda)}\\
&=\int_0^\infty d \zeta \int_0^\zeta d j\, g(\zeta-j)S(j)je^{\lambda j}\frac{e^{\lambda(\zeta-j)}}{Z(\lambda)}\\
&=\int_0^\infty S(j)je^{\lambda j}dj\,,
\end{align*}
where on the right hand side we have that $\lambda=\lambda(\rho)$, which is the inverse function of \eqref{robeta} and as before we have a finite value only for $\lambda<\tilde\lambda_c$ for a suitable critical value $\tilde \lambda_c$.
Therefore, the diffusion coefficient is given by
\begin{align}\label{diff}
D(\rho)=\frac{d}{d\rho}\hat D(\rho)&=\int_0^\infty S(j)j^2e^{\lambda(\rho) j}\lambda'(\rho)dj=\lambda'(\rho)\sigma(\rho)
\end{align}
and the Einstein relation \eqref{ein} is satisfied.

\smallskip

In the following theorem, we will completely characterize the integrable zero range models, i.e., models having rates of the form \eqref{zrrates} (or the equivalent form in the discrete case) and having transport coefficients satisfying either condition \eqref{remkp2} or condition \eqref{remk}.

Since the state space of the model allows for only positive values of the $\eta$, we should also obtain only positive densities $\rho\geq 0$ in the scaling limit. This implies that the mobility has to be such that $\sigma(0)=0$. Otherwise, it is easy to construct a suitable external field and a suitable initial condition
that give a solution of the hydrodynamic equation which is negative at some point. For this reason, in the next theorem we only consider polynomials that have a zero at zero. The case when the mobility is nonnegative in an interval $[l,r]$ with $l>0$ and $\sigma(l)=\sigma(r)=0$ can be reduced to a case with the mobility zero at zero with a shift of $-l$ in the mass on each vertex. However, in the case of a zero range dynamics this situation cannot happen, because for a nonnegative mobility in a bounded interval $[l,r]$ we need to have that the mass on the sites of the lattice stays bounded between the values $l<r$, but this is not possible since the rate at which the mass is transferred depends only on the amount of mass present at the starting site and arbitrarily large values of the mass can be attained. Likewise, by shifting the mass we can reduce the case when the mobility is positive on $[l,+\infty)$ to the case when the mobility is positive on $[0,+\infty)$.

In the sequel we consider therefore $\mathcal P_2(\rho)=b\rho(\rho+a)$ as the generic second order polynomial in \eqref{remkp2}. For the zero range dynamics, from the previous argument, we can only have $b>0$ and $a\geq 0$. We will denote by $\displaystyle{\binom{r}{k}=\frac{r(r-1)\cdots (r-k+1)}{k!}}$, with $r\in \mathbb R$ and $k\in \mathbb N$, the generalized binomial coefficient and by $\Gamma(z)=\int_0^{+\infty}t^{z-1}e^{-t}\,dt$ the Euler gamma function.
\begin{remark}
In the following Theorem we characterize $g$ and $S$ as measures; in the absolutely continuous case the rates
\eqref{zrrates} are absolutely continuous and obtained by the ratio of the densities, in the atomic case the mass is quantized and only integer multiple of the mass $Ma$ (to be defined in the next Theorem) may jump. In this case the rates \eqref{zrrates} are atomic and the weight associated to the rate of $jMa,$ $j\in \mathbb N$ is given by the products of the corresponding weights for the measures $g$ and $S$, as will be specified in \eqref{ilflussozr}.
\end{remark}
\begin{theorem}\label{teozr}
Consider a zero range model whose rates and invariant measures are determined by the functions $g$ and $S$ as in Examples~\ref{zre} and \ref{zre2}.
We have the following:
\begin{itemize}
\item[1)]
Given $a\geq 0$ and $b, c >0$, we have that $D(\rho)=c$ and $\sigma(\rho)=b\rho(\rho+a)$ if and only if one of the possible two cases happen.

\smallskip

\noindent 1A) The first case is when $a>0$; in this case $g$ and $S$ are atomic measures and
\begin{align*}
g=Re^{-aK}\sum_{j=0}^{\infty}\left|\binom{-\frac{1}{M}}{j}\right|e^{-KMaj}\delta (x-Maj)
\end{align*}
with
\begin{align*}
S=\frac{c^2}{b}\sum_{j=1}^{\infty}\frac{e^{-KMaj}}{j}\delta(x-Maj),
\end{align*}
for some arbitrary $R, K\in\mathbb R^+$ and $M=\frac{b}{c}$.

\smallskip

\noindent 1B) The second case is when $a=0$ and then
\begin{equation}\label{g2}
g(x)= \frac{Rx^{\frac{1-M}{M}}e^{-Kx}}{\Gamma(1/M)}\,,
\end{equation}
with
\begin{equation}\label{S2}
S(x)=\frac{be^{-Kx}}{M\Gamma(2)x}\,,
\end{equation}
for some arbitrary $R, K\in\mathbb R^+$ and $M=\frac{b}{c}$.

\vspace{+5pt}
\item[2)] Given $c\in\mathbb R$ we have that $D(\rho)=c\sigma'(\rho)$ if and only if
\begin{align}
S(dx)=2Kc^2\delta\left(x-\frac{1}{2c}\right)\,,
\end{align}
for some $K\in\mathbb R^+$; while instead $g$ is arbitrary.
\end{itemize}
\begin{proof}
We first prove item 1). By \eqref{mob} and \eqref{diff}, in order to have $D(\rho)=c$ and $\sigma(\rho)=b\rho(\rho+a)$, the following conditions have to be satisfied
\begin{align}\label{sys}
\begin{cases}
\displaystyle{\int_0^\infty S(y)y^2e^{\lambda(\rho)y}dy=b\rho(\rho+a)},\\
\displaystyle{b\rho(\rho+a)\lambda'(\rho)=c}.
\end{cases}
\end{align}
By the second equation in \eqref{sys} we get
\begin{align}\label{dife}
\lambda'(\rho)=\frac{c}{b\rho(\rho+a)}.
\end{align}
By solving \eqref{dife} we get the following expression for $\lambda(\rho)$
\begin{equation}
\lambda(\rho)=\left\{
\begin{array}{ll}
\frac{c}{ab}\log\left(\frac{\rho}{\rho+a}\right)+K & a>0\,,\\
-\frac{c}{b\rho} +K & a=0\,,
\end{array}
\right.
\end{equation}
for some $K\in\mathbb R$.
Therefore if we call $M=\frac{b}{c}$ its inverse function is
\begin{equation}\label{rho2}
\rho(\lambda)=\left\{
\begin{array}{ll}
\frac{ae^{Ma\lambda}}{e^{MaK}-e^{Ma\lambda}}\,, & a>0\\
\frac{1}{M(K-\lambda)} & a=0\,.
\end{array}
\right.
\end{equation}
Since
$\displaystyle{\rho(\lambda)=\frac{d}{d\lambda}\left(\log Z(\lambda)\right)}$, integrating we have
\begin{equation}
Z(\lambda)=\left\{
\begin{array}{ll}
R\exp\left( \int\frac{ae^{Ma\lambda}}{e^{MaK}-e^{Ma\lambda}}d\lambda \right)\,, & a>0\,,\\
\frac{R}{(K-\lambda)^{\frac 1M}}\,, & a=0\,,
\end{array}
\right.
\end{equation}
for some constant $R\in\mathbb R^+$.

\smallskip
Let us first discuss the case $1A)$ when $a>0$.
By defining $y:=e^{Ma\lambda}$ we get
\begin{align}\label{zeta}
Z(\lambda)&=R\exp\left( \frac{1}{M}\int\frac{1}{e^{MaK}-y}dy      \right)\\
&=R\exp\left( -\frac{1}{M}\log(e^{MaK}-y) \right)\\
&=\frac{R}{(e^{MaK}-e^{Ma\lambda})^{\frac{1}{M}}}\,.
\end{align}
By the definition of $Z(\lambda)$, relation \eqref{zeta} is equivalent to
\begin{align}\label{gi}
\int_0^\infty g(x)e^{\lambda x}dx=\frac{R}{(e^{MaK}-e^{Ma\lambda})^{\frac{1}{M}}}.
\end{align}
Therefore we obtain that $g(x)$ has to be the inverse Laplace's transform of the right hand side of \eqref{gi}. Since
\begin{align}\label{usoancora}
\frac{R}{(e^{MaK}-e^{Ma\lambda})^{\frac{1}{M}}}=Re^{-Ka}\sum_{j=0}^{\infty}
\left|\binom{-\frac{1}{M}}{j}\right|e^{Ma(\lambda-K)j},
\end{align}
we get that $g$ satisfies \eqref{gi} if and only if it is an atomic measure and
\begin{align}
g(dx)=Re^{-aK}\sum_{j=0}^{\infty}\left|\binom{-\frac{1}{M}}{j}\right|e^{-KMaj}\delta \Big(x-Maj\Big)\,.
\end{align}
Moreover, replacing the first relation in \eqref{rho2} into the first equality of the system \eqref{sys}, we get
\begin{align}\label{inva}
\int_0^\infty S(y)y^2e^{\lambda y} dy&=a^2b\frac{e^{MaK}e^{Ma\lambda}}{(e^{MaK}-e^{Ma\lambda})^2}\\
&=a^2b e^{MaK}e^{Ma\lambda}\sum_{j=0}^{\infty}\left|\binom{-2}{j}\right|e^{Ma\lambda j}e^{-MaK(2+j)},
\end{align}
which implies that the unique measure that satisfies \eqref{inva} is the inverse Laplace's transform of the right hand side of \eqref{inva}, which means
\begin{align*}
S(dx)=\frac{b}{M^2}\sum_{j=1}^{\infty}\frac{1}{j}e^{-MaKj}\delta\Big(x-Maj\Big)\,,
\end{align*}
where we used the identity $\left|\binom{-2}{j-1}\right|=j$, $j\geq 1$.

\smallskip

We consider now the case $1B)$, that is $a=0$. We obtain that
\begin{equation}
\int_0^{+\infty}g(x)e^{\lambda x}dx=\frac{R}{(K-\lambda)^{\frac 1M}}\,,
\end{equation}
and the inverse Laplace transform of the right hand side of the above formula can be computed and it is given by
\begin{equation}
g(x)= \frac{Rx^{\frac{1-M}{M}}e^{-Kx}}{\Gamma(1/M)}\,.
\end{equation}
Considering the first equation in \eqref{sys} with $a=0$ and the second equation in \eqref{rho2} we obtain
\begin{equation}
\int_0^{+\infty}S(y)y^2 e^{\lambda y}dy= \frac{b}{M(K-\lambda)^2}\,,
\end{equation}
and again the inverse laplace transform can be computed and gives
\begin{equation}
S(x)=\frac{be^{-Kx}}{M\Gamma(2)x}\,.
\end{equation}
In both cases the inverse Laplace transforms are corresponding to classic cases and can be found on textbooks.

\smallskip

We now prove item 2). By \eqref{diff}, it holds that $D(\rho)=c\sigma'(\rho)$ if and only if
\begin{equation}
\frac{d}{d\rho}\pare{\log\sigma(\rho)}=\frac{1}{c}\lambda'(\rho).
\end{equation}
Therefore
\begin{align}\label{sig}
\sigma(\rho)=Ke^{\frac{\lambda(\rho)}{c}},
\end{align}
for some constant $K\in\mathbb R^+$. By \eqref{mob} and \eqref{sig} we get
\begin{equation}
\int_0^\infty S(dj)j^2e^{\lambda(\rho)j}dj=Ke^{\frac{\lambda(\rho)}{c}}
\end{equation}
which implies
\begin{align}
S(dx)=Kc^2\delta\left(x-\frac{1}{c}\right)\,.
\end{align}
Since $\lambda(\rho)$ is arbitrary the function $g$ is arbitrary too.
\end{proof}
\end{theorem}

\begin{remark}
We observe that it is not possible to construct a zero range dynamics in which the mobility is equal to a constant $K\in\mathbb R$. Indeed, by \eqref{mob}, the condition $\sigma(\rho)=K$ is satisfied either if $\lambda(\rho)$ is constant or if $S(dx)=\delta(x)$ is a delta measure at zero. In both cases we have trivial dynamics.
\end{remark}

\noindent Recall the identities $\binom{x}{y}=\frac{\Gamma(x+1)}{\Gamma(y+1)\gamma(x-y+1)}$ and $\frac{\Gamma(s-a+1)}{\Gamma(s-b+1)}=(-1)^{b-a}\frac{\Gamma(b-s)}{\Gamma(a-s)}$.

\noindent Item $1A)$ in Theorem \ref{teozr} identifies a class of discrete zero range models that are integrable in the sense described in Section \ref{sec:int}. The mass is quantized and only a multiple of the mass $Ma$ can jump in a single step. Starting initially with masses that are multiple of $Ma$ this property is preserved therefore in time. By formula \eqref{zrrates} we have that the rate at which a mass equal to $Maj$, $j\in \mathbb N$ is jumping from one site containing a mass equal to $\eta(x)Ma$, to a nearest
neighbor site is given by
\begin{equation}\label{ilflussozr}
Q(\eta(x)Ma,jMa)=\frac{\Gamma\left(\eta(x)-j+\frac 1M\right)\Gamma\left(\eta(x)+1\right)}{\Gamma\left(\eta(x)+\frac 1M\right)\Gamma\left(\eta(x)-j+1\right)}\frac {c^2}{bj}\,, \qquad 1\leq j\leq \eta(x)\,.
\end{equation}
As a special case, we consider $c= 1/(2s)$, $b=1/(4s^2)$, $a=2s$ for a suitable parameter $s>0$ so that $M=1/(2s)$. We obtain that the mass that can jump must be an integer number (so that we have particles with unitary mass jumping) with the rate \eqref{ilflussozr} becoming
\begin{equation}\label{giardina}
Q(\eta(x),j)=\frac{\Gamma\left(\eta(x)-j+2s\right)\Gamma\left(\eta(x)+1\right)}{\Gamma\left(\eta(x)+2s\right)\Gamma\left(\eta(x)-j+1\right)}\frac {1}{j}\,, \qquad 1\leq j\leq \eta(x)\,,
\end{equation}
where  $\eta(x)$ and $j$ are integers.
Formula \eqref{giardina} gives the same rates considered in \cite{FG} where the authors obtain a boundary driven interacting particle system that consists of a finite chain of $N$ sites connected at its endpoints to two reservoirs and whose non-equilibrium steady state has a closed form \cite{FG,CFFGR,CFGGT}.
Hence, our result shows that the models in \cite{FG} are, modulo generalization to non integer mass jumping, the only discrete zero range dynamics that have the macroscopic integrability in \ref{sec:int} of the stationary non equilibrium state .

\begin{remark}\label{rem:transf}
Note that, even if we constructed a large class of models, depending on several parameters, we can without loss of generality restrict to the case of particles having unitary mass. Indeed given a particle system with particles of unitary mass, we can construct in a trivial way another system with particles of mass $\mu$ just by considering the same jumping rate related to the number of particles, and changing the mass. Using the formulas for the transport coefficients we can directly check that the new transport coefficients $D^\mu$, $\sigma^\mu$ are related to the unitary one by the relation
\begin{equation}\label{transformation}
\left\{
\begin{array}{l}
D^\mu(\rho)=D\left(\frac{\rho}{\mu}\right)\,,\\
\sigma^\mu(\rho)=\mu^2\sigma\left(\frac{\rho}{\mu}\right)\,.
\end{array}
\right.
\end{equation}
The integrability conditions are preserved inside the class of all those models obtained varying $\mu$. Since the rates \eqref{ilflussozr}, once fixed the constants, depend only on the occupation numbers $\eta(x)$ and $j$ we have that all models in $1A)$ are naturally obtained in this way. For this reason, in the sequel we will mainly consider the discrete unitary case.
\end{remark}

\smallskip

On the other hand, item $1B)$ identifies instead a class of solvable zero range dynamics with a continuous distribution of the mass flowing. Indeed recalling \eqref{zrrates} we obtain macroscopically integrable zero range models with rates of jump given by
\begin{equation}
Q(\eta(x),j)=\frac{b}{M\Gamma(2)}\frac 1j \left(1-\frac{j}{\eta(x)}\right)^{\frac{1-M}{M}}\,;
\end{equation}
the construction and definition of models with these rates is delicate and we do not discuss here.

\smallskip

Finally item $2)$ gives models that even if considered in contact with reservoirs have a large deviations rate functional for the invariant measure that is local, i.e., there are no long range correlations. The result of our Theorem is that this situation may happen only if we have a classic zero range model in which the mass can take values that are multiple of a fixed amount and at each jump of the process just one particle can jump.

\section{Other examples}\label{sec:other}
In this section we briefly discuss some integrable models beyond the zero range case. A more general treatment is possible but it would require a much longer discussion that it is not possible to insert here.

\subsection{Generalized KMP models}
We start with a class of models that are gradient and reversible with respect to the measure \eqref{mis} when the function $g(\zeta)$ is constant, say $1$, and for which the integrability conditions of Section~\ref{sec:int} hold. In the continuous case we have therefore exponential invariant measures while in the discrete case we obtain geometric distributions. Furthermore, in view of Remark~\ref{rem:transf}, for the discrete case we consider models for which the elementary unit of mass that may jump is one. This
is because of the validity of the transformation \eqref{transformation}.
Note also that both properties \eqref{remk} and \eqref{remkp2} are preserved by this transformation.

\begin{example}[Generalized discrete KMP model]
We discuss now a particular case of Example~\ref{ex1}.
Let us fix $k>0$ and $c\in\mathbb R$ such that $k-\frac{|c|}{2}\geq 0$ and consider the case in which the rate $R(\zeta, j)$ is defined as in \eqref{rataj} when $h(j)=k$ and $\displaystyle{A\pare{\zeta, a}=\frac{ca}{\zeta+1}}$ which is a generalized version of the discrete KMP model that is obtained when $k=\frac{1}{2}$ and $c=1$.
By a direct computation, distinguishing several cases, we have that the condition $k-\frac{|c|}{2}\geq  0$ guarantees that
\begin{align}\label{rate}
\inf_{\zeta_1, \zeta_2\in\mathbb N}\inf_{j\in [-\zeta_2, \zeta_1]}R(\zeta_1,\zeta_2, j)\geq 0.
\end{align}

By Example \ref{ex1} we know that the gradient and the reversibility conditions are satisfied. The instantaneous current is given by
\begin{align}\label{isc}
j_\eta(x,y)&=\sum_{j=1}^{\eta(x)}j R\pare{\eta(x), \eta(y), j}-\sum_{j=1}^{\eta(y)}jR\pare{\eta(y), \eta(x), j}\\
&=k(\eta(x)-\eta(y)).
\end{align}
When the function $g$ is constant and the measure defined in \eqref{mis} is atomic,
the normalization factor is $Z(\lambda)=\sum_{\zeta=0}^\infty e^{\lambda \zeta}=\frac{1}{1-e^\lambda}$ and the density: $\rho(\lambda)=\frac{e^\lambda}{1-e^{\lambda}}$. Therefore, its inverse function is $\lambda(\rho)=\log\frac{\rho}{1+\rho}$.
The marginal of the measure $g_\lambda$, written in terms of $\rho$, has the following expression:
\begin{align}\label{meas}
g_{\lambda(\rho)}(\zeta)=\pare{\frac{\rho}{\rho+1}}^\zeta \frac{1}{\rho+1}.
\end{align}
By definition of $\lambda(\rho)$ we have that
\begin{equation*}
\hat D(\rho)=\mathbb E_{g_{\lambda(\rho)}}\left( k\zeta\right)=k\rho.
\end{equation*}
Therefore, $D(\rho)=\frac{d}{d\rho} \hat D(\rho)=k$,
while the mobility coefficient is given by
\begin{align}\label{sigma}
\sigma(\rho)=&\frac 12\mathbb E_{g_{\lambda(\rho)}\otimes g_{\lambda(\rho)} }\pare{\sum_{j=-\zeta_2}^{\zeta_1}j^2 R\pare{\zeta_1, \zeta_2, j}}.
\end{align}
Observe that
\begin{align}
\sum_{j=-\zeta_2}^{\zeta_1}j^2 R\pare{\zeta_1, \zeta_2, j}&=\sum_{j=0}^{\zeta_1}j^2 R\pare{\zeta_1, \zeta_2, j}+\sum_{j=0}^{\zeta_2}j^2 R\pare{\zeta_1, \zeta_2, j}\\
&=\sum_{j=0}^{\zeta_1}\pare{k-\frac{c}{2}}j+\sum_{j=0}^{\zeta_1}\frac{c}{\zeta_1+\zeta_2+1}j^2\\
&\hspace{+10pt}+\sum_{j=0}^{\zeta_2}\pare{k-\frac{c}{2}}j+\sum_{j=0}^{\zeta_2}\frac{c}{\zeta_1+\zeta_2+1}j^2\\
&=\pare{k-\frac{c}{2}}\pare{\frac{\zeta_1(\zeta_1+1)}{2}+\frac{\zeta_2(\zeta_2+1)}{2}}\\
&\hspace{+10pt}+\frac{c}{\zeta_1+\zeta_2+1}\pare{\frac{\zeta_1(\zeta_1+1)(2\zeta_1+1)}{6}+\frac{\zeta_2(\zeta_2+1)(2\zeta_2+1)}{6}}.
\end{align}
Since
\begin{align*}
\hspace{-15pt}\frac{\zeta_1(\zeta_1+1)(2\zeta_1+1)}{6}+\frac{\zeta_2(\zeta_2+1)(2\zeta_2+1)}{6}=&\frac{1}{3}\pare{\zeta_1+\zeta_2+1}\pare{\zeta_1^2+\zeta_2^2-\zeta_1\zeta_2}+\frac{1}{6}\pare{\zeta_1+\zeta_2+1}^2\\
&-\frac{1}{6}\pare{\zeta_1+\zeta_2+1},
\end{align*}
we get that
\begin{align}
\sum_{j=-\zeta_2}^{\zeta_1}j^2 R\pare{\zeta_1, \zeta_2, j}=&
\pare{k-\frac{c}{2}}\pare{\frac{\zeta_1(\zeta_1+1)}{2}+\frac{\zeta_2(\zeta_2+1)}{2}}\\
&+\frac{c}{3}\pare{\zeta_1^2+\zeta_2^2-\zeta_1\zeta_2}+\frac{c}{6}\pare{\zeta_1+\zeta_2+1}-\frac{c}{6}\\
=&\pare{\frac{k}{2}+\frac{c}{12}}\pare{\zeta_1^2+\zeta_2^2}+\pare{\frac{k}{2}-\frac{c}{12}}\pare{\zeta_1+\zeta_2}-\frac{c}{3}\zeta_1\zeta_2.
\end{align}
By \eqref{sigma} and \eqref{meas}, we get
\begin{align}\label{mobi}
\sigma(\rho)=&\frac{1}{2(\rho+1)^2}\sum_{\zeta_1=0}^\infty\sum_{\zeta_2=0}^\infty \pare{\frac{k}{2}+\frac{c}{12}}\pare{\zeta_1^2+\zeta_2^2}\pare{\frac{\rho}{\rho+1}}^{\zeta_1+\zeta_2}\\
&+\frac{1}{2(\rho+1)^2}\sum_{\zeta_1=0}^\infty\sum_{\zeta_2=0}^\infty \pare{\frac{k}{2}-\frac{c}{12}}\pare{\zeta_1+\zeta_2}\pare{\frac{\rho}{\rho+1}}^{\zeta_1+\zeta_2}\\
&-\frac{1}{2(\rho+1)^2}\sum_{\zeta_1=0}^\infty\sum_{\zeta_2=0}^\infty\frac{c}{3}\zeta_1\zeta_2\pare{\frac{\rho}{\rho+1}}^{\zeta_1+\zeta_2}.
\end{align}
Since, for all $\lambda\in(0,1)$, it holds that
\begin{align}\label{ds}
&\sum_{\zeta_1=0}^\infty\sum_{\zeta_2=0}^\infty\zeta_1^2\lambda^{\zeta_1+\zeta_2}=\frac{2\lambda^2}{(1-\lambda)^4}+\frac{\lambda}{(1-\lambda)^3},\\
&\sum_{\zeta_1=0}^\infty\sum_{\zeta_2=0}^\infty\zeta_1\lambda^{\zeta_1+\zeta_2}=\frac{\lambda}{(1-\lambda)^3},\label{ds2}\\
&\sum_{\zeta_1=0}^\infty\zeta_1\lambda^{\zeta_1}=\frac{\lambda}{(1-\lambda)^2},  \label{ds3}
\end{align}
we get
\begin{align}\label{bb}
\sigma(\rho)=&\frac{1}{(\rho+1)^2} \pare{\frac{k}{2}+\frac{c}{12}}\pare{2(\rho+1)^2\rho^2+(\rho+1)^2\rho}\\
&+\frac{1}{(\rho+1)^2} \pare{\frac{k}{2}-\frac{c}{12}}(\rho+1)^2\rho-\frac{1}{(\rho+1)^2}\frac{c}{3}(\rho+1)^2\rho^2\\
&=k(\rho^2+\rho).
\end{align}
\end{example}

\begin{example}[Generalized continuous KMP model] We discuss now a continuous version of the previous example.
In the exponential case in which the measure $g_\lambda(d\zeta)$ defined in \eqref{mis} is absolutely continuous with respect to the Lesbegue measure and $g(\zeta)=1$, we compute the normalization factor $Z(\lambda)=\int_0^\infty e^{\lambda\zeta} d\zeta=\frac{-1}{\lambda}$, $\lambda <0$, and the density $\rho(\lambda)=-\frac{1}{\lambda}$. Its inverse function is $\lambda(\rho)=\frac{-1}{\rho}$ and the marginal of the measure $g_\lambda$, written in terms of $\rho$, has the following expression
\begin{align*}
g_{\lambda(\rho)}(d\zeta)=\frac{1}{\rho}e^{-\frac{\zeta}{\rho}}d\zeta.
\end{align*}

We fix $k\in\mathbb R^+$ and $c\in\mathbb R$ such that $k-\frac{|c|}{2}\geq 0$ and consider the case in which $h(j)=k$ and $A\pare{s, y}=\frac{cy}{s}$. Therefore, $\forall j>0$,
\begin{align}\label{ratec}
R\pare{\zeta, j}={\frac{k}{j}+\frac{c}{j\pare{\zeta_1+\zeta_2}}\pare{j-\frac{\zeta_1+\zeta_2}{2}}}=\pare{k-\frac{c}{2}}\frac{1}{j}+\frac{c}{\zeta_1+\zeta_2}.
\end{align}
Observe that the KMP is obtained for $k=\frac{1}{2}$ and $c=1$.
Condition $k-\frac{|c|}{2}\geq 0$ guarantees that
$\displaystyle{
\inf_{\zeta_1\geq 0, \zeta_2\geq 0}\inf_{j\in [-\zeta_2, \zeta_1]}R(\zeta_1, \zeta_2, j)\geq 0.}$
By Example \ref{exco} we know that the dynamics is gradient and reversible with respect to the measure $g_\lambda(d\zeta)$.
The instantaneous current is given by
\begin{align*}
j_\eta(x,y)&=\int_{-\eta(y)}^{\eta(x)}jR\pare{\eta(x), \eta(y), j}dj=k\eta(x)-k\eta(y)\,.
\end{align*}

By \eqref{trc} the coefficient $\hat D(\rho)=\mathbb E_{g_{\lambda(\rho)}}\pare{k\zeta}=k\rho$ and therefore $D(\rho)=\frac{d}{d\rho}\hat D(\rho)=k$, while the mobility coefficient is given by
\begin{align*}
\sigma(\rho)=\frac 12\mathbb{E}_{g_{\lambda(\rho)}\otimes g_{\lambda(\rho)}}\pare{\int_{-\zeta_2}^{\zeta_1}R(\zeta,j)j^2dj}.
\end{align*}
Observe that
\begin{align*}
\int_{-\zeta_2}^{\zeta_1}R(\zeta,j)j^2dj
&=\int_{0}^{\zeta_1}\pare{k-\frac{c}{2}}j\,dj-\int_{-\zeta_2}^0\pare{k-\frac{c}{2}}jdj+\int_{-\zeta_2}^{\zeta_1}{\frac{cj^2}{\zeta_1+\zeta_2}}dj\\
&=\pare{k-\frac{c}{2}}\pare{\frac{\zeta_1^2}{2}+\frac{\zeta_2^2}{2}}+\frac{c}{3(\zeta_1+\zeta_2)}\pare{\zeta_1^3+\zeta_2^3}\\
&=\pare{k-\frac{c}{2}}\pare{\frac{\zeta_1^2}{2}+\frac{\zeta_2^2}{2}}+\frac{c}{3}\pare{\zeta_1^2+\zeta_2^2-\zeta_1\zeta_2}
\end{align*}
Therefore
\begin{align*}
\sigma(\rho)&=\frac 12\pare{\frac{k}{2}+\frac{c}{12}}\int_0^\infty\int_0^\infty                              \pare{{\zeta_1^2}+{\zeta_2^2}}
\frac{e^{-\frac{\zeta_1+\zeta_2}{\rho}}}{\rho^2}d\zeta_1d\zeta_2\\
&\hspace{+10pt}-\frac{c}{6}\int_0^\infty\int_0^\infty\zeta_1\zeta_2
\frac{e^{-\frac{\zeta_1+\zeta_2}{\rho}}}{\rho^2}d\zeta_1d\zeta_2\\
&=\frac 12\pare{k-\frac{c}{6}}\rho^2.
\end{align*}
We remark once again that and the  structure of the fluctuations for the above model is not straightforward since we are in a case on which the total rate of jump across a bond is infinite so that the general theory for classic particle systems cannot be applied directly.

\end{example}

\subsection{Misanthrope type process}
We briefly consider some cases of the models in Section~\ref{sec:mis}.
We consider rates of the form \eqref{RGPR}; substituting in \eqref{sbaglirip} we obtain that $H(\zeta)=-\frac{k}{\hat S(\zeta)}$ and we recall also that $g=\hat S$. Let us call
$m_1(\lambda)=\int_0^{\infty} S(j)j e^{\lambda j} dj$ and $m_2(\lambda)=\int_0^{\infty} S(j)j^2 e^{\lambda j} dj$.
Considering  $\lambda<0$ we have
\begin{equation}\label{nn}
Z(\lambda)=\int_0^{\infty}\hat S(\zeta)e^{\lambda\zeta}d\zeta=
-\frac{1}{\lambda}\int_0^{\infty} S(j)j e^{\lambda j} dj=-\frac{m_1(\lambda)}{\lambda}\,.
\end{equation}
Applying formulas of section \ref{sec:TC}, the mobility coefficient is given by:
\begin{eqnarray*}
	\sigma(\rho) & = & \frac 12\mathbb E_{g_{\lambda(\rho)}\otimes g_{\lambda(\rho)} }\left(\int_{-\zeta_2}^{\zeta_1}R(\zeta_1,\zeta_2,j)j^2 dj\right)\\
	& = &
	k
	\int_0^{\infty}\int_0^{\zeta_1}S(j)j^2
	\frac{e^{\lambda\zeta_1}}{Z(\lambda)}\left(\int_0^{\infty}
	\frac{e^{\lambda\zeta_2}}{Z(\lambda)}d\zeta_2\right) dj d\zeta_1\\
	& = &
	k\frac{-1}{\lambda Z(\lambda)}
	\int_0^{\infty}S(j)j^2
	\left(\int_j^{\infty}
	\frac{e^{\lambda\zeta_1}}{Z(\lambda)}
	d\zeta_1\right) dj\\
	& = &
	k\frac{1}{(\lambda Z(\lambda))^2}
	\int_0^{\infty}S(j)j^2 e^{\lambda j} dj,\\
	& =& \frac{m_2(\lambda)}{m_1(\lambda)^2}\,,
\end{eqnarray*}
where $\lambda=\lambda(\rho)$.
Similarly, for the diffusion coefficient we have:
\begin{equation}
\hat D(\rho) =\mathbb E_{g_{\lambda(\rho)}}(h(\eta))=
-k\int_0^{\infty}\frac{e^{\lambda\zeta}}{Z(\lambda)} d\zeta
=
\frac{k}{\lambda Z(\lambda)}=-\frac{k}{m_1(\lambda)}\,.
\end{equation}
The validity of the Einstein relation \eqref{ein} follows directly since we have
$\hat D=G(m_1)$ and $\sigma=G'(m_1)m_2$ where
$G$ is the function $G(x)=-k/x$. For the zero range dynamics we have the same mechanism but $G(x)=x$.

\smallskip
In this case we may have a confinement of the mass in a bounded interval and consequently we may have integrable models with a quadratic but concave mobility. Therefore, we have to consider, in addition to the polynomials considered in Theorem \ref{teozr}, also the case on which the parameters $a,b$ are negative.
The analysis of the second formula in \eqref{sys} in Theorem \ref{teozr} does not depend on the specific model and can be done also in this case. We have however now that the expansion \eqref{usoancora} is a finite one in the case that $-1/M\in \mathbb N$; this corresponds to models that have a threshold in the number of particles, like for example the exclusion process. We do not discuss a complete characterization that it is difficult in this case.

\section*{Declarations}
The authors have no conflicts of interest. Data sharing is not applicable to this article as no datasets were generated or analysed during the current study.

\end{document}